\documentclass[10pt,twocolumn,a4paper]{article}
\usepackage{color}
\usepackage{amsmath}
\usepackage{yhmath}
\usepackage{cite}
\usepackage{algorithm}
\usepackage{algorithmic}
\usepackage{amssymb}
\usepackage{graphicx}
\usepackage{setspace}
\usepackage{amsbsy}
\usepackage{amsthm}
\usepackage{tabularx}
\usepackage{bbm}
\usepackage{appendix}
\newtheorem{prop}{Proposition}
\newtheorem{prb}{Problem}

\begin{document}
\title{ {\huge Cache Placement Optimization in Mobile Edge Computing Networks with Unaware Environment - An Extended Multi-armed Bandit Approach}}

\author{{Yuqi Han, Rui Wang, Jun Wu, Dian Liu, and Haoqi Ren}
\thanks{Y. Han, R. Wang, D. Liu, and H. Ren are with the College of Electrical and information Engineering,
Tongji University, Shanghai, 201804 China, J. Wu, is with Scool of Computer Science, Fudan University, Shanghai, 201203, China, Emails: yqhan@tongji.edu.cn, ruiwang@tongji.edu.cn, wujun@tongji.edu.cn, liudian622@163.com, renhaoqi@tongji.edu.cn.}
}

\maketitle

\begin{abstract}

Caching high-frequency reuse contents at the edge servers in the mobile edge computing (MEC) network omits the part of backhaul transmission and further releases the pressure of data traffic. However, how to efficiently decide the caching contents for edge servers is still an open problem, which refers to the cache capacity of edge servers, the popularity of each content, and the wireless channel quality during transmission.
In this paper, we discuss the influence of unknown user density and popularity of content on the cache placement solution at the edge server. Specifically, towards the implementation of the cache placement solution in the practical network, there are two problems needing to be solved.
First, the estimation of unknown users' preference needs a huge amount of records of users' previous requests.
Second, the overlapping serving regions among edge servers cause the wrong estimation of users' preference, which hinders the individual decision of caching placement. To address the first issue, we propose a learning-based solution to adaptively optimize the cache placement policy without any previous knowledge of the user density and the popularity of the contents. We develop the extended multi-armed bandit (Extended MAB), which combines the generalized global bandit (GGB) and Standard Multi-armed bandit (MAB), to iteratively estimate both a global parameter, i.e., the user density, and individual parameters, i.e., the popularity of each content. For the second problem, a multi-agent Extended MAB based solution is presented to avoid the mis-estimation of parameters and achieve the decentralized cache placement policy. The proposed solution determines the primary time slot and secondary time slot for each edge server.  The edge servers estimates expected satisfied user number of caching a content with the overlap information and determine the cache placement solution.
The proposed strategies are proven to achieve the bounded regret according to the mathematical analysis. Extensive simulations verify the optimality of the proposed strategies when comparing with baselines.

\end{abstract}

\section{Introduction}

{

Emerging 5G networks provides the high-efficiency transmission under challenge of the increasing amount of data traffic and user number, which brings economic and technical benefits. Specifically, mobile edge computing (MEC) is introduced in the 5G network to shorten the back-haul transmission distance by deploying part of tasks at the edge server.
Caching popular data contents at the edge serves enables part of users' requests to be directly responded to by the edge servers rather than by the faraway data center, which reduces the delivery latency and greatly improves the efficiency of the wireless transmission.
However, because of the limitation of the cache capacity, the edge servers cannot cache all contents. In this paper, we aim to find the optimal cache placement policy at edge servers to maximize the number of satisfied user requests by caching contents. There are two challenges we address in the paper.

First, the cache placement solution depends on the popularity of each content, which is hard to acquire for edge servers.  Meanwhile, the user number is dynamic as time goes by and unknown. With the unknown user density, the number of requesting a content at a time slot cannot directly reflect the popularity of the content. Second, to cover all mobile devices in the wireless network, some of the edge servers are placed with the overlapped serving region. In this case, devices in the overlapped serving region could be served by multiple edge servers.  The optimal cache placement solution not only depends on the practical popularity of the content but also the cache placement of the adjacent edge servers that have the overlapping region with it.

Towards the two problems, we propose a learning-based cache placement solution without any prior knowledge of the user density and the content popularity. Moreover, the proposed cache placement solution could avoid inaccurate estimation of parameters and performance loss under the overlapped region when multiple edge servers have overlapping serving region.

For the first problem, we conduct the cache placement optimization with the multi-armed bandit (MAB) model under the unknown user density and popularity of each content. The standard MAB, which regards each arm individual, and Generalized global bandit (GGB), which assumes the arms have a known relationship with an unknown parameter, cannot handle the two different unknown parameters concurrently. To fill the gap between the practical requirement of diverse parameter estimation and constraints of the MAB models, we propose the Extended MAB for cache placement optimization, which considers both the individual parameter and the global parameter. The user density and popularity of each content are estimated concurrently, which improves the learning efficiency of the optimal cache placement solution.

For the second problem, we first propose a centralized cache placement policy, in which one of the dimensions represents all cache placement choices, and the other one covers all edge servers,  to eliminate the inaccuracy of  parameter estimation. Moreover, to reduce the complexity, we propose a multi-agent Extended-MAB for edge servers to individually make cache placement. Each edge server in the network is regarded as an agent and has its exclusive time slot to make parameter estimation. With the information of the overlapped region, the edge server estimates the expected serving user number and further derives the optimal cache placement solution.}

The main contributions of this paper are summarized as follows:
\begin{itemize}
\item We propose a cache placement solution under unknown user density and popularity of content. To the best of our knowledge, this is the first work that considers both of user density and popularity of content unknown in the cache placement solution. Specifically, we propose the Extended MAB, which is first used to deal with the cache placement in an \emph{individual edge server scenario} where edge servers own a non-overlapping serving region.  The regret is discussed to specify its theoretical performance.

\item To avoid the miscalculation due to the overlapped serving region in a large-scale MEC network, the proposed Extended-MAB is modified to deal with the joint caching placement optimization in a \emph{cooperative edge server scenario}. A dedicated edge server manager is first set to collect overall network information and output the cache placement action combination for all edge servers. Moreover, a multi-agent Extended-MAB in which the edge servers individually perform actions is proposed to release the burden of computation complexity overhead of the joint cache placement optimization.

\item  We conduct a series of experiments with different parameters and network settings. We choose $5$ different baselines and $3$ metrics to verify the availability of the proposed cache placement solution. According to the experiment, the proposed cache placement solution shows the best performance in all the scenarios.
\end{itemize}

\section{Related works}
In this section, we introduce existing research relevant to the cache placement policy and MAB. In what follows, we mainly review the research of the cache placement policy under the unknown environment of the MEC network. After that, we review some extension studies of the standard MAB theory and discuss the applications, which are relevant to the proposed extended MAB model in this paper. We also emphasize the difference between the mentioned research and proposed solution in this paper.

\subsection{Cache placement with unknown environment of MEC network}

The cache placement policy in the MEC network is influenced by the wireless channel condition, such as the signal to noise ratio (SNR), and the attributes of mobile users, such as the user density and popularity of requesting contents.

\cite{Zhou_2017, Zhang_2017} investigate the cache placement optimization with unknown channel conditions.  Since the noise and fading of wireless channel influence the transmission error rate, unknown characteristics of wireless channel conditions may decrease the successful transmission probability or waste the bandwidth.  In particular, \cite{Zhou_2017} models a cache-enabled network and assumes the channel state as a Markov decision process. The value function and state-action cost are discussed in different channel cases. \cite{Zhou_2017} discusses the optimal cache scheduling under the dynamic wireless network. \cite{Zhang_2017} aims to maximize the spectrum efficiency under the unknown channel condition and make the cache placement decision under the group-based structure. \cite{Zhang_2017} divides the entire wireless cell into several groups to avoid interference in the D2D network and proposes the primal-dual adaptive cache placement algorithm.

Towards the unknown popularity of contents, some pieces of researches consider predicting the popularity of contents using the machine learning method for the cache placement solution. For example, \cite{Yin_2018, Liu_2018} employ the deep learning-based prediction by firstly collecting users' requests as the training data.  In \cite{Nagaraja_2015, Bharath_2016}, transfer learning is introduced to estimate the popularity, which is further used to design the cache placement strategy. Specifically, \cite{Nagaraja_2015} studies the cache placement in a heterogeneous network where the content popularity information is unaware. \cite{Bharath_2016} discusses content correlation and information transition between periods and uses the auto-regressive (AR) model to predict the users' requests. The above model-based algorithms need the assumption about the users' requests, which is hard to be obtained in the practical network.

{
The exploration and exploitation trade-off is further discussed to avoid the prior assumption of the popularity of the content for the cache placement solution.  \cite{Blasco_2014} introduces the content controller to learn the unknown popularities of contents by observing the instantaneous demand from users and discusses the relationship between the cache placement solution and the factors, such as the number of files, the number of users, the cache size, and the skewness of the popularity profile. \cite{Song_2017} and \cite{Jiang_2019} both consider using a multi-armed bandit for the cache placement problem. \cite{Song_2017} applies the semidefinite relaxation approach in the centralized cache placement situation, where the cache placement strategy of all BSs is jointly derived. Besides, \cite{Song_2017} also provides a distributed algorithm such that each base station could make their own decision. \cite{Jiang_2019} indicates that the cache placement is influenced not only by the popularity of each content but also by the users' preference. Because of the unknown users' preference, the proposed model observes the historical content and derives the cooperative cache content strategy. {It is noted that
\cite{Blasco_2014,Song_2017,Jiang_2019} do not discuss the influence of unknown user density to the cache placement solution.  Specifically,  \cite{Song_2017,Jiang_2019} assume the fixed user number in the wireless environment and optimize the cache placement solution. However, the unknown user density will slow down the learning efficiency because the popularity of content is calculated based on the number of users' requests. In this paper, we consider the cache placement under the dynamic environment. At each time slot, the user number in the wireless network changes following a random distribution. We propose the Extended MAB to learn the user density and popularity concurrently. With the Extended MAB, the cache placement solution under the unknown user density and popularity could be learned faster, which further improves the efficiency of cache placement solution.}

How the overlapped region among multiple edge servers affect the caching design is still an open problem. Little research focuses on the how to make full use of the benefit from the multiple edge servers responding the requests.
\cite{Chattopadhyay_2018} assumes part of base stations know about the popularity of content and proposes Gibbs sampling-based method based on the knowledge of contents stored in its neighboring base stations. This method cannot be conducted at a environment with totally unknown popularity of content. In \cite{Xu_2020}, the number and location of mobile users are defined in advance and fixed. However, the assumption of the pre-determined location and requests of mobile users is still unrealistic in the practical network.
{In this work,  we propose a general cache placement strategy without the determined number and locations of users and maximize the number of satisfied users in the whole region.
We discuss the influence of the mis-calculation of the user density and popularity, which is caused by multiple edge servers responding the requests. To solve the problem, we propose a time-division method for addressing the problem introduced by the mis-calculation of user density and popularity. }}

\subsection{Multi-armed bandit}

\cite{Mersereau_2008} points out that the independence of arms in the multi-armed bandit causes the increase of convergence time with a large number of arms. To solve the problem, \cite{Mersereau_2008} introduces a bandit with the mean reward of arms following a linear function.
\cite{Rusmevichientong_2010} proposes a multi-dimensional linear bandit and the cumulative Bayes risk under an arbitrary policy is at least $\mathcal{O}(\sqrt{T})$, where $T$ denotes the running time of the proposed algorithm. The complexity is lower than $\mathcal{O}(\log{T})$ in a standard multi-armed bandit problem.

\cite{Atan_2015} proposes a Global Multi-armed Bandit (GMAB), in which arms are globally informative through a global parameter. The rewards in GMAB follow different distributions but with the same parameter. The GMAB model has fewer constraints on the reward distribution, hence can cover more scenarios. \cite{Shen_2018} extends the GMAB to GGB to handle nonmonotonic but decomposable reward functions, multidimensional global parameters, and switching costs. The proposed greedy algorithm converges to the optimal arm in a finite time period. According to the results, the proposed algorithm significantly outperforms existing bandits solutions. \cite{Wang_2018} proposes group-based multi-armed bandits by combining the standard multi-armed bandit and the GGB. Moreover, the paper proposes Upper Confidence Bound-greedy (UCB-g) algorithm to solve the regional bandit model. The proposed strategy achieves order-optimal regret by exploiting the intra-region correlation and inter-region independence.

Combinatorial MAB (CMAB) is proposed in \cite{Chen_2013}, which allows the agents to play a set of arms at each time. By minimizing the ($\alpha$-$\beta$) approximation regret, the online CMAB could converge to the optimal solution. The CMAB   achieves $\mathcal{O}(log n)$ distribution-dependent regret after $n$ times playing. \cite{Jiang_2019, Jiang_20192, Song_2017} all refer to the structure of the CMAB to solve the cache placement problems in different MEC models.

In this paper, we take the benefits from the CMAB and propose the Extended MAB, which admits choosing multiple arms at each time, to solve the cache placement in the MEC network.
The main difference between the proposed Extended MAB and other MABs lies in that the proposed Extended-MAB includes both the individual parameter and global parameter.

\section{System model}\label{sm}

We consider a cache-enabled MEC network with $M$ edge servers and use $m$ to denote the index of the edge server.
There are $N$ contents in the network and the contents form the content set $\mathcal{N},\mathcal{N}\triangleq \left\{1,2,...,N\right\}$. In this paper, we consider the data size of each content is the same $\footnote{If the data size of each content is different, we divide them into small pieces with same size and regard them as multiple individual contents.}$.
We define $m$ could cache $K$ contents and the set of caching contents at edge server $m$ forms an combination $\mathbbm{I}_{m,t}$. We use $i_{m,t}$ to denote the index of caching contents of edge server $m$ at $t$.  The number of combination is denoted as $C$, $C = {{N}\choose{K}}$.

{
In this paper, we study the influence of unknown users' preferences on the cache placement solution and ignore the attributes of wireless transmission.
We assume that edge servers could successfully transmit content to the users in its serving region. At each time slot $t$, the edge servers broadcast the caching contents to all users in its serving region, and the caching content could be successfully received. In this case, if the requested contents are cached at edge servers, the requests are directly satisfied without needing to fetch the contents from the central server.} 
{The edge server only needs to receive a signal to indicate whether the user device is satisfied by the caching content. Based on the satisfied user number, the edge server estimates the user density and popularity of contents to further optimize the cache placement solution.

We investigate two different scenarios of the distribution of edge servers as illustrated in Fig.~\ref{indcahing} and Fig.~\ref{coopcaching}. Fig.~\ref{indcahing} presents an \emph{individual edge server scenario} in which the serving region of each edge server does not overlap with other edge servers. The yellow circle in Fig.~\ref{indcahing} indicates the serving region of an edge server. If the user requests a cached content, the edge server is capable of transmitting the content without backhaul transmission as the green lines showing. Otherwise, the central server transmits them to the user. The corresponding transmissions are denoted by red lines. In Fig.~\ref{coopcaching}, we discuss the \emph{cooperative edge server scenario}, where the users have the chance of being served by multiple edge servers. As illustrated in Fig.~\ref{coopcaching}, the users in region 1 and region 3 are served by their corresponding edge servers, but the users in region 2 could be served by both two edge servers. If the user in the overlapped location requires a content, it takes the benefit from the cache space so long as one of the edge servers caches this content.  In this case, if all edge servers choose the most popular content to cache, the overall cache placement policy may only achieve a sub-optimal performance.   
We define the location of each edge severs is deterministic, the information of the overlapped region is assumed to be known by all edge servers.

\begin{figure}[t]
\centering
\includegraphics[width=0.5\textwidth]{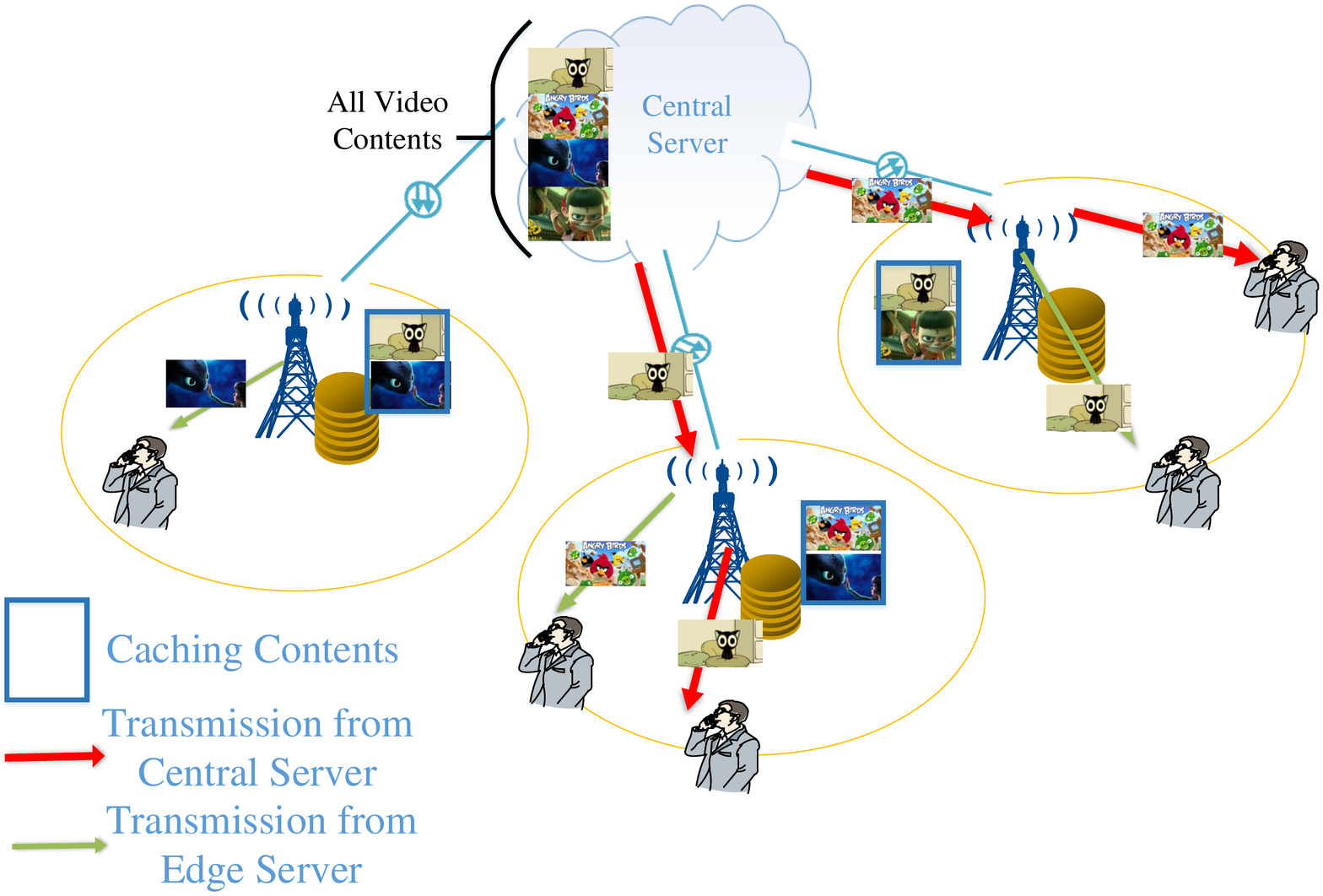}
\caption{Illustration of MEC network with non-overlapped serving region.}
\label{indcahing}
\end{figure}

\begin{figure}[t]
\centering
\includegraphics[width=0.5\textwidth]{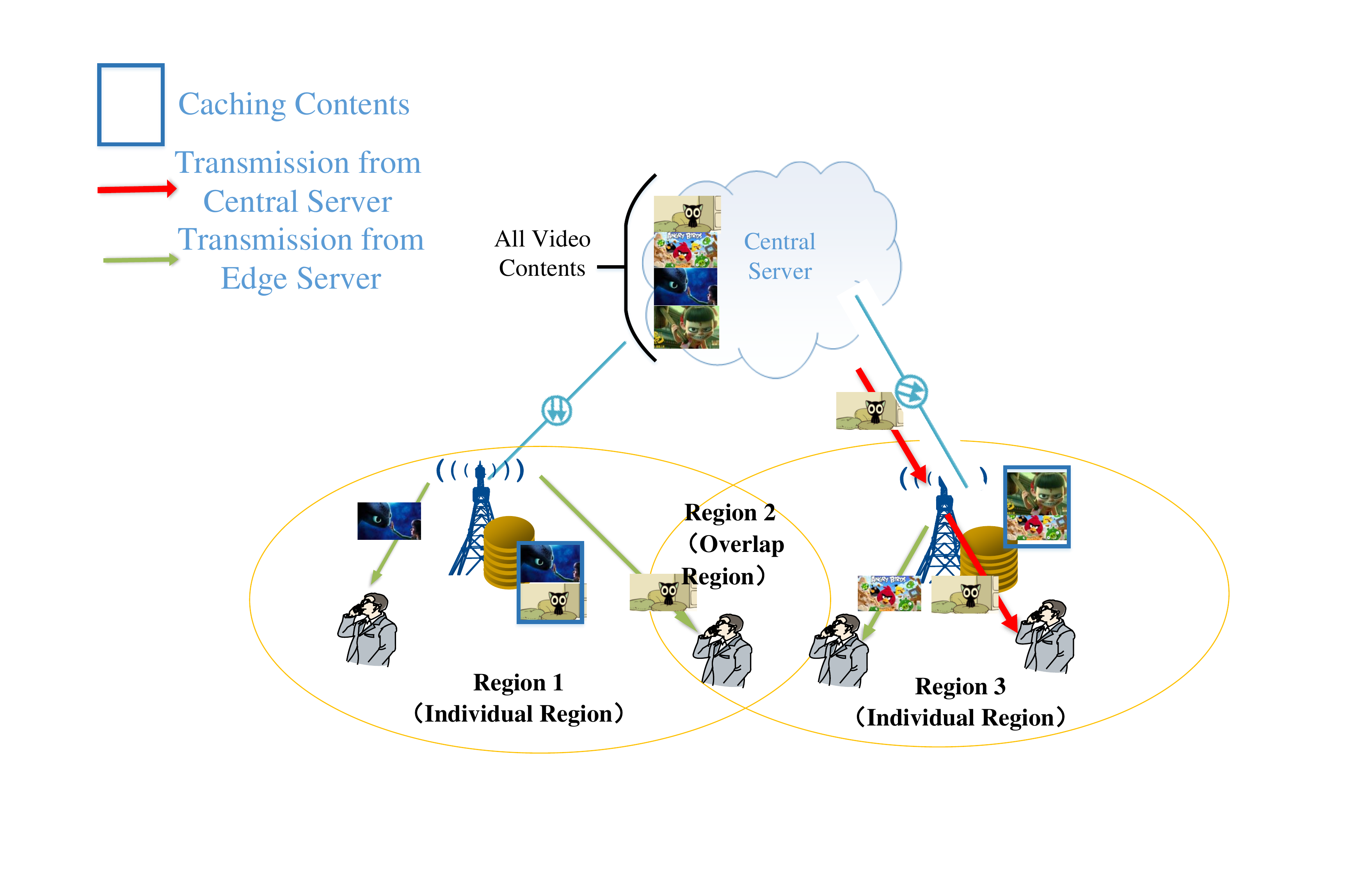}
\caption{Illustration of MEC network with overlapped serving regions.}
\label{coopcaching}
\end{figure}

We define the serving region of $m$ as $R_m$ and the overall region of the MEC network is $R^*$. The relationship between the serving region of edge server $m$ and the whole region is
\begin{equation}
\forall m, R_m\leq R^*\leq\sum_{m = 1}^M R_m.
\end{equation}
When the serving region of all edge server are non-overlapped, $R^*=\sum_{m = 1}^M R_m$ holds.}

We consider users distributed in the network following Poisson point process (PPP) at each time slot $t$. The user density reflects the ratio of actual user number $u^*_t$ to the area size of serving region, which is denoted as $v_{t}, v_{t} = \frac{u^*_t}{R^*}$.
We assume that the user density follows a certain distribution $v_{t}(\theta)$ with expectation $\mu(\theta)$, where $\theta$ is unknown.   Moreover, the expectation of user density function satisfies the following conditions\cite{Atan_2015}.

\begin{itemize}
\item For $\theta,\theta'\in\Theta$, there exists $D_1> 0$ and $0<\gamma_1<1$ such that $|\mu^{-1}(\theta) - \mu^{-1}(\theta') | \leq D_1|\theta-\theta'|^{\gamma_{1}}$, where $\mu^{-1}$ is the inverse function of $\mu$.

\item For $\theta,\theta'\in\Theta$, there exists $D_2> 0$ and $0<\gamma_2\leq1$, such that $|\mu(\theta) - \mu(\theta')|\leq D_2|\theta - \theta'|^{\gamma_2}$
\end{itemize}

Specifically, in this paper, we apply function $\mu(\theta) = w\theta^k + b$ as the expectation of user density. The exponential parameter $k$ reflects the influence from the variable $\theta$. The parameter $w$ and $b$ could be changed to adapt to the different environment.

The user number of requesting content $n$ at $t$ in the overall region is $u_{n,t}, \sum_{n=1}^N u_{n,t} = u^*_t$.
Without loss of generality, we assume that the popularity of the contents are $p_1, p_2,...,p_N$ with $p_1\ge p_2\ge ...\ge p_N$ and $\sum_{n=1}^N p_n =1$.
Under the above the description of the popularity of contents and the user density, we could derive the amount of requiring content $n$ denoted by $u_{n,t}= p_n u^*_t$. It is noted that both the parameter of user density $\theta$ and the popularity  $p_n$ of certain content $n$ are unknown in a practical MEC network. The optimal cache placement of each edge server depends on the accurate estimation of the parameters.

{
When users' requests are satisfied by the cached content at the edge server, the user device sends a signal to inform the edge server. The edge servers aim to maximize the number of serving users by the caching contents, which is written as $u_{\mathbbm I_{m,t}}$.}
The problem is formulated as
\begin{prb}
\label{prb1}
\begin{equation}
\begin{split}
&\max\lim_{T\to\infty}\sum_{t = 1}^T \sum_{m=1}^{M} u_{\mathbb I_{m,t}}\\
&=\max\lim_{T\to\infty}\sum_{t = 1}^T \sum_{m=1}^{M}  \sum_{i_{m,t}\in \mathbb I_{m,t}}p_{i_{m,t}}v_{t}(\theta) R_{m}\\
&\mbox{s.t.}\quad (1),\quad\quad |\mathbb I_{m,t}|\leq K.
\end{split}
\end{equation}
\end{prb}

According to problem~\ref{prb1}, the formulated optimization problem is difficult to solve as the user density and preference are unknown. In this paper, we refer to the MAB model to learn the unknown parameter  and further adjust the cache placement solution online, which overcomes the weakness of traditional method in the environment with unknown parameters. We propose the Extended MAB, which combines the attributes of the MAB and the GGB and provides a more efficient way to determine the cache placement in the MEC network in scenarios of both individual edge servers and cooperative edge servers.  {The edge server estimates the expected parameters at each time slots based on the number of the user satisfied by cached contents and does not need to record the user number for fetching each specific content.}

\section{Extended Multi-armed Bandit}

In this section, we present the Extended MAB for following cache placement optimization, which incorporates the global parameter and the individual local parameters concurrently. 

{We first briefly introduce the model of standard MAB. The MAB is proposed to solve the exploration and exploitation dilemma in the unknown environment optimization.  Given the countable arms, i.e., the optional actions, the agent chooses one of the arms as action at each time.  A reward is returned after the arm is taken. In the beginning, the agent does not know the reward of each arm and explores the environment by randomly choosing the arms and acquires the rewards. With the accumulation of knowledge of the reward of each arm, the agent could choose the optimal arm to maximize the summary of the reward. If the agent chooses the best-estimated arm too early, the loss of the reward may occur because of the lack of knowledge of the environment. However, if the agent always chooses an arm randomly, it can not make full use of the knowledge of environment and derive the optimal action.

At each time $t$, the agent takes $a_t$ and acquire the reward $r_{a_t}$, the objective is written as 
\begin{equation}
\lim_{T\to\infty}\max_{a_t}\sum_{t = 1}^T r_{a_t}.
\end{equation}

In the GGB, the reward is determined by the reward distribution function $v_{a_t}(\theta)$, whose expectation is $\mu_{a_t}(\theta)$. The global parameter $\theta$ is shared by all arms. The agent chooses the arm and acquires the knowledge of the parameter $\theta$. If $\theta$ is accurately estimated, the optimal arm could be chosen directly without any other exploration. The objective in GGB is written as
\begin{equation}
\begin{split}
&\lim_{T\to\infty}\max_{a_t}\sum_{t = 1}^T r_{a_t}(\theta)\\
=&\lim_{T\to\infty}\max_{a_t}\sum_{t = 1}^T\mathbb{E}(v_{a_t}(\theta))\\
=&\lim_{T\to\infty}\max_{a_t}\sum_{t = 1}^T\mu_{a_t}(\theta).
\end{split}
\end{equation}

The Extended MAB to integrate the global parameter and the individual parameter of each arm in one bandit model. The reward of each arm depends on two types of parameters, which is denoted as $r_{a_t} = p_{a_t} v(\theta)$. We use $\mu(\theta)$ to denote the expectation of distribution function $v(\theta)$. The objective of Extended MAB is to maximize the accumulated reward, which is
\begin{equation}
\begin{split}
&\lim_{T\to\infty}\sum_{t = 1}^T\max_{a_t} p_{a_t} \mathbb{E}(v({\theta}))\\
=&\lim_{T\to\infty}\sum_{t = 1}^T\max_{a_t} p_{a_t} \mu({\theta}).
\end{split}
\end{equation}

The knowledge of global parameter could be obtained no matter which cache placement policy is chosen, while the individual parameter can only be obtained when this content is cached.}


\section{Cache Placement in Individual Edge Server Scenario}
In the individual edge server scenario, all edge servers have their exclusive serving regions and make independent cache placement.  Given the constraints of the cache capacity $K$, the edge server $m$ chooses a cache combination $\mathbbm I_{m,t}$.  The $u_{\mathbbm I_{m,t}}$ indicates the total amount of all satisfied users by the caching combination $\mathbbm{I}_{m,t}$ at time $t$, which can be written as
\begin{equation}
u_{\mathbbm I_{m,t}}= \sum_{i_{m,t}\in\mathbbm I_{m,t}} u_{i_{m,t}}.
\end{equation}

We transform Problem 1 as maximizing the average number of satisfied user's requests in the infinite time duration, which is presented as
\begin{equation}
\begin{split}
&\lim_{T\to\infty}\sum_{t = 1}^T\max_{\mathbbm I_{m,t}}  u_{\mathbbm I_{m,t}}\\
= &\lim_{T\to\infty}\sum_{t = 1}^T\max_{\mathbbm I_{m,t}} \sum_{i_{m,t}\in\mathbbm I_{m,t} }\mathbb{E}( u_{i_{m,t}}).
\end{split}
\end{equation}

The user number of requesting a content depends on the user density in the serving region and the popularity of the content. To maximize the satisfied user number by the caching content, we use the Extended MAB to learn the cache placement policy.  We define the cache combination $\mathbbm{I}_{m,t}$ as the arm that the agent chooses. The processs is divided into 3 parts, which are initialization, exploration and exploitation, and parameter estimation. The details are shown as follows.

\subsubsection{Initialization}
At the beginning of the algorithm, the edge server $m$ does not know the information of contents and caches nothing. The users in the region send requests to $m$ and $m$ fetches the requesting contents from the central server.  The estimated parameters of user density $\hat\theta$ and popularity $\hat p_n$ of each content $n$ are initialized by $0$.  We use $B$ to denote the size of a batch, in which the edge server makes the same cache placement policy.

\subsubsection{Exploration and exploitation} The trade-off between exploration and exploitation follows a determined rule. If time $t$ satisfies $\log_2(t*B)\in\mathbb{N}$, the random cache placement combination is chosen. Otherwise, the edge server chooses the best combination according to the estimated parameters. With this policy, when the parameters are correctly estimated, the policy decreases the randomness and makes cache placement decisions according to the estimated parameters.

\subsubsection{Parameters estimation}
{The expected reward $\bar X_{\mathbbm{I}_{m,t}}$ of combination $\mathbbm{I}_{m,t}$ until time $t$ is updated after choosing it, which is calculated based on the previously acquired reward and current reward. We use $\bar X_{\mathbbm{I}_{m,t}}^*$ to denote the updated expected reward of combination $\mathbbm{I}_{m,t}$. Once the combination $\mathbbm{I}_{m,t}$ is chosen, the expected reward of $\mathbbm{I}_{m,t}$ at time $t$ is updated as
\begin{equation}
\bar X_{\mathbbm{I}_{m,t}}^* = \frac{M_{\mathbbm{I}_{m,t}}(t-1)\bar X_{\mathbbm{I}_{m,t}} +  X_{\mathbbm{I}_{m,t}}}{M_{\mathbbm{I}_{m,t}}(t-1) + 1}
\end{equation}
where $M_{\mathbbm{I}_{m,t}}(t-1)$ denotes the number of choosing $\mathbbm{I}_{m,t}$ until time $t-1$.
}
Since the combination with content $n$ has $K-1$ remaining space to cache, the combination with $n$ has ${N-1}\choose{K-1}$ different choices. Hence, the combination set, which includes content $n$, is composed by ${N-1}\choose{K-1}$ combinations. If we sum up the reward of all combinations, we have
\begin{equation}
\sum_{c_t= 1}^{C} \bar X_{c_t} ={{N-1}\choose {K-1}} \sum_{n_t=1}^{N}\bar X_{n_t}
={{N-1}\choose {K-1}} \mu(\hat\theta).
\end{equation}

The estimated global parameter $\hat \theta$ is derived from the sum of expectation rewards of all combinations. Based on Eq. (9), the expected global parameter is given as
\begin{equation}
\hat\theta = \arg\min_{\theta\in\Theta} \bigg{|}{{N-1}\choose{K-1}}* \mu(\theta) - \sum_{c_t= 1}^{C} \bar X_{c_t}\bigg{|}.
\end{equation}

The popularity of each cache placement combination $c$ could be derived by
\begin{equation}
\hat p_c = \bar X_{c}/(\mu(\hat \theta)).
\end{equation}

{{Until now, we could derive the estimated popularity of each combination, which supports choosing the cache placement solution. Hence, we do not need to know the estimated popularity of each individual content.}} The details of the process are given in Algorithm 1.
\begin{algorithm}[]
\caption{Cache Placement for Individual Edge Server Scenario Based on Extended MAB} 
\begin{algorithmic}
\STATE Initialize the cache size $K$; the number of combinations $C = {N \choose K}$; the distribution function of user density $\mu(\theta)$; $\hat\theta = 0$; $\hat a_c = 0$; $M_{c} = 0$, $t = 1$; the batch size $B$;
\WHILE{$t\ge1$}
\STATE{$b =0$}
\WHILE{$b<B$}
\IF{$\log_2(t) \in\mathbb{N}$}
\STATE Select combination $\mathbb{I}_{m,t}$ randomly for set $\mathcal{C}$;
\ELSE
\STATE Select combination $\mathbb{I}_{m,t}$  which satisfies $\mathbb{I}_{m,t} = \arg\max_{c,t} \hat p_{c,t} \mu(\hat\theta)$;
\ENDIF
\STATE $\bar X_{c_t} = \bar X_{c_t} $ for $c\in\mathcal{C}\backslash \mathbb{I}_{m,t}$;
\STATE Update $\bar X_{\mathbb{I}_{m,t}}$ with (8);
\STATE Update $\hat\theta$ with (10);
\STATE Update $\hat p_c$ with (11);
\STATE $M_{\mathbb{I}_{m,t}}(t) = M_{\mathbb{I}_{m,t}}(t-1) +1$;
\STATE $b = b+1$;
\ENDWHILE
\ENDWHILE
\end{algorithmic}
\end{algorithm}

We show the regret analysis and complexity analysis in the following part of the section. We indicate the proposed Extended MAB based cache placement solution approaches to optimal cache placement policy.

\begin{prop}
In the individual edge server scenario, the proposed Extended MAB converges to the expected regret $(\frac{\log T}{T}* B + 2\exp\big(-2({{N-1}\choose{K-1}}\frac{\sigma_1}{\bar D_{1}})^{2\gamma_1} T\big) + 2N\exp \big(-2(\frac{\sigma_2}{D_2}\mu(\theta^*)^{2\gamma_2} T) \big))*gap_{max}$
at time $T$,
where parameters $D_1$, $D_2$, $\gamma_1$, $\gamma_2$ and $\sigma_1$, $\sigma_2$ satisfy $D_1>0$,$D_2>0$, $0<\gamma_1<1$, $0<\gamma_2<1$, $\sigma_1>0$, $\sigma_2>0$, $gap_{max}$ denotes the maximal value gap between the optimal reward and any other rewards.
\end{prop}

\begin{proof}

To find the optimal $K$ arms, we need to reduce the gap between the practical maximum reward and the evaluated maximum reward as much as possible. We define the chosen arms indexed by $I_1,I_2,...,I_K\in\mathbb{I}_{m,t}$. The regret of expected reward between the ideal best cache placement and the estimated cache placement action at each time slot is denoted as
\begin{equation}
\begin{split}
   Reg =\sum_{k =1}^K a_k\mu(\theta) - a_{I_k}\mu(\theta).
\end{split}
\end{equation}

We analyze the expected regret $Reg(T)$ to prove that the Extended MAB could find the best cache placement, which is written as
\begin{equation}
Reg(T) = \frac{1}{T}\sum_{t = 1}^{T} \sum_{k = 1}^{K} (a_k\mu(\theta) - X_{I_{k,t}}).
\end{equation}

We separately discuss the regrets of the algorithm with exploration action and exploitation action. We firstly calculate the regret of random cache placement.  The regret of random cache placement is defined as $Reg_{1}(T)$ , which satisfies
\begin{equation}
\begin{split}
Reg_{1}(T) &\leq \frac{\log_2 T}{T}* B*gap_{max}*K
\end{split}
\end{equation}

The proof of Eq.(14) is shown as appendix~\ref{ap1}.

The regret with the exploitation action $Reg_{2}(T)$ is written as
\begin{equation}
\begin{split}
Reg_{2}(T)\leq& (2\exp\big(-2({{N-1}\choose{K-1}}\frac{\sigma_1}{\bar D_{1}})^{2\gamma_1} T\big)\\+&2N\exp \big(-2(\frac{\sigma_2}{D_2}\mu(\theta^*)^{2\gamma_2} T) \big))*gap_{max}.
\end{split}
\end{equation}

The proof of Eq.(15) is shown as appendix~\ref{ap2}.

Since $p_n<1$ is always satisfied, by summing up $Reg_{1}(T)$ and $Reg_{2}(T)$, the total regret of Algorithm 1 is written as
\begin{equation}
\begin{split}
Reg(T) \leq& \frac{\log T}{T}*B*gap_{max}*K \\
+&2\exp\big(-2({{N-1}\choose{K-1}}\frac{\sigma_1}{\bar D_{1}})^{2\gamma_1} T\big)*gap_{max} \\
+&2N\exp \big(-2(\frac{\sigma_2}{D_2}\mu(\theta^*)^{2\gamma_2} T) \big)*gap_{max}.
\end{split}
\end{equation}

The estimated optimal cache placement is taken in the exploitation policy.
\end{proof}

     {The complexity of the Extended MAB based cache placement solution in individual scenario depends on the number of contents and the capacity of cache size. At each iteration, estimating the parameters requires computing the expected reward of all combinations, the complexity of which is $\mathcal{O}({{N}\choose{K}})$. While finishing the estimation of parameters, the edge server searches the cache placement solution with the highest expected reward. The complexity of searching algorithm is $\mathcal{O}({{N}\choose{K}})$. Hence, the complexity of the Extended MAB based cache placement solution is $\mathcal{O}({{N}\choose{K}})$.}

\section{Cache Placement in Cooperative Edge Server Scenario}

In this section, we discuss the cache placement in cooperative edge server scenarios with overlapped serving regions.  We first draw the lesson from Algorithm 1 and propose a centralized Extended MAB algorithm with the assistance of the central server. Then, to reduce the size of the action space of the centralized algorithm, we propose a decentralized multi-agent Extended MAB algorithm, where the edge server makes decisions individually but achieves the global
optimal cache placement with low computational complexity.

\subsection{Centralized Extended MAB based cache placement solution}
In a real large-scale MEC network,  the users in the overlapped region could receive the requesting content from one or more edge servers.  If all edge servers choose the most popular contents to cache, the users in the overlapped regions will lose the extra chance to be satisfied. To make full use of the cache space, we firstly propose a centralized cache placement solution, in which a central server is introduced to make cache placement policy. The scope of the serving region and the overlapped region is aware by the central server in advance.

The cache placement combination of each edge server is regarded as a sub-combination. The central server chooses $M$ sub-combinations for $M$ edge servers. Assembling all $M$ sub-combinations $\mathbbm{I}_{m,t}$ produces  a macro-combination $C_t$. The reward is determined by the number of satisfied users in the global network. The popularity of contents and user density are iteratively estimated according to the accumulated reward from the beginning to the present. The cache placement based on the centralized Extended MAB bandit is described as follows.

\subsubsection{Initialization} Given the cache size $K$ and the amount of the total contents $N$. The  macro-combination space $\mathcal{C}$ has $C$ elements, where $C={N\choose K}^M$. The reward is recorded in the $\bar{X}_{c}$ for macro-combination $c$. The parameters of the user density and the popularity of each macro-combination are initialized by $0$.

\subsubsection{Exploration and exploitation} The policy of balancing exploration and exploitation is same as the situation with individual edge server. If $t$ satisfies $\log_2(t)\in\mathbb{N}$, the central server randomly chooses the macro-combination. Otherwise, the central server chooses the estimated best macro-combination. It is noted $B$ in the cooperative scenario should be set larger than the individual scenario because of more choice of macro-combination.

\subsubsection{Parameter estimation}

The parameter of the user density is estimated after obtaining the reward at each time slot. Later, we estimate the popularity of each combination under the estimated user density. 
We firstly determine how many times that content $n$ is shown in all macro-combinations. As there are ${N-1}\choose{K}$ sub-combinations without content $n$ , we have ${{N-1}\choose{K}}^M$ macro-combinations are in the environment excluding the content $n$. Hence we obtain that ${{N}\choose{K}}^M-{{N-1}\choose{K}}^M$ macro-combinations contain combination $n$. The user density is derived as
\begin{equation}
 \mu(\hat\theta) = \frac{\sum_{c= 1}^C \bar X_{c_t}}{{{N}\choose{K}}^M-{{N-1}\choose{K}}^M}.
\end{equation}

After that, the popularity of each macro-combination could be derived according to Eq. (11).

\subsection{Decentralized Extended MAB based cache placement solution}
{
With the increase of the edge servers' amount, the space of the cache placement combination grows exponentially. When adding an extra content, the number of cache placement combinations is increased by ${{N+1}\choose{K}}^M - {{N}\choose{K}}^M$, which makes it difficult to be implemented in a practical large-scale wireless network. To solve this problem, we introduce a decentralized framework and let edge servers make their own cache placement solution based on the proposed Extended MAB.

There are two issues to be addressed in the cooperative edge server scenario as they do not occur in the individual edge server case. First, the overlapped serving region leads to a miscalculation of the parameters. We take a case with 2 edge servers sharing the overlapped serving region as an example and denote two edge servers as $B_0$ and $B_1$. When the edge servers $B_0$ and  $B_1$ both cache content $n$, the users requesting content $n$ in the overlapped region are satisfied only by one of them. The reward received by the edge servers, which reflects the actual satisfied user number, is always no more than the reward in the individual case.

To avoid the mis-estimation, we propose a time-division parameter estimation framework, in which each edge server is assigned a dedicated primary time slot and multiple sharing secondary time slots. Edge servers update the parameters only when they are in the primary time slot. Since there are $M$ edge servers in the network, $M$ time slots are set as a group. In each group,  the edge servers have  $1$ exclusive primary slot and $M-1$ shared secondary slot. The proposed time-division framework is shown in Fig.~\ref{td}. The users in the overlapped region tend to be served by the edge server in the primary slot. Based on this scheme, if  $B_0$ and $B_1$ cache the same content $n$ and $B_0$ is in the primary state slot,  $B_0$ has priority to receive the response from the users in the overlapped region, which reflects the actual satisfied user number of the overlapped region. Then $B_0$ updates the parameters according to the reward. Meantime, $B_1$ neglects the parameter estimation of policy adjustment and only broadcasts all caching contents.

\begin{figure}[t]
\centering
\includegraphics[width=0.5\textwidth]{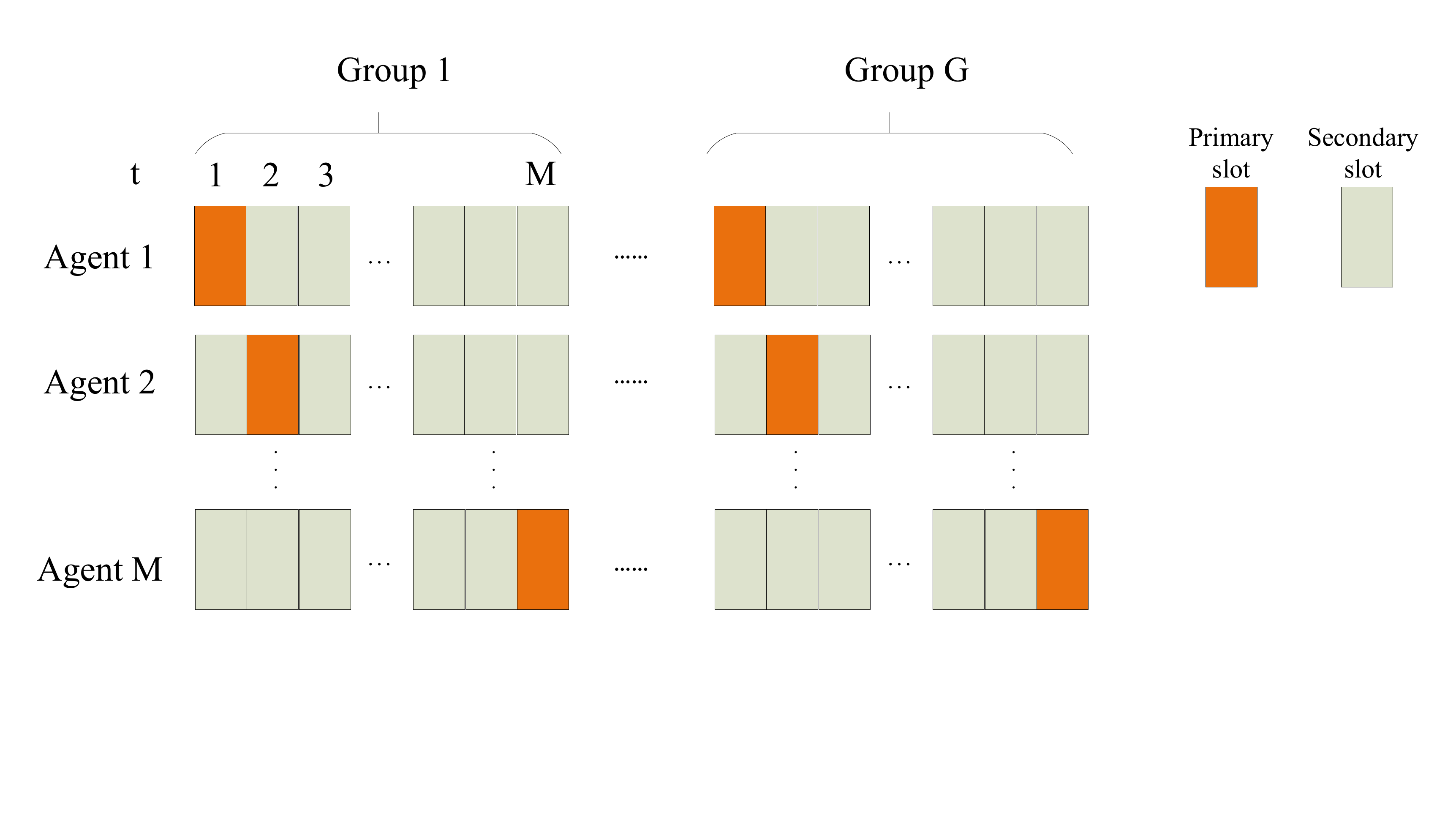}
\caption{Time division framework for multi-agent Extend MAB.}
\label{td}
\end{figure}

Second, the best cache placement of each edge server is determined not only by the estimated user density and popularity of the content but also by the cache placement of adjacent edge servers.  Under the accurately estimated parameters, the edge server cannot directly choose the cache placement solution with maximal expected reward because the edge servers which has overlapped region with it may have the same choice.

We assume each of $M$ edge servers have overlapped region with at least $1$ other edge server.
In the decentralized cache placement solution, to find the best cache placement strategy,  the edge servers transmit the cache placement solution to adjacent edge servers after it makes decisions in the primary time slot.  With the cache placement solution and size of the overlapped region,  the edge servers derive the popularity of each content under the exclusive primary time slot. Given the popularity  $p_{\mathbbm{I}_{m,t}}$ of each combination $\mathbbm{I}_{m,t}$, we first calculate the the popularity of combinations including $n$, denoted by $s_n$, as
\begin{equation}
    s_n = \sum_{n\in \mathbbm{I}} \bar p_\mathbbm{I}.
\end{equation}

It is noted that there are $\tbinom{N-1}{K-1}$ combinations including $n$, and $\tbinom{N-2}{K-2}$ combinations with other contents. Hence, the estimated popularity $\hat p_n$ is given by
\begin{equation}
   \hat p_n = \frac{s_n - \tbinom{N-2}{K-2}}{\tbinom{N-1}{K-1} - \tbinom{N-2}{K-2}}.
\end{equation}

{
The serving region is divided into different sub-regions based on the overlap with other edge servers.  Since the location of each edge server is predefined, we could identify edge servers that each sub-region belongs to and calculate the area size of each sub-region.} Under the estimated popularity of each content, we choose the content for $M*K$ cache units at all edge servers. 

\begin{prop}
We order the estimated popularity $p_n$ of each content $n$ from highest to lowest and use $f_i$ denote the index after the ordering. We define $S$ as the best cache placement set, $s\triangleq\left\{f_1,f_2,...,f_{M*K}\right\}$. In other word, the $M*K$ content with highest estimated popularity forms the best cache placement set $S$.
The contents of optimal cache placement solution for the cooperative edge server scenario always belong to best cache placement set $S$.
\end{prop}
\begin{proof}
We define the optimal cache placement solution of edge servers as $i^*$ and the obtained reward  is $r^*$.
Under the accurate estimation, if the content $n, n\in i^*$ does not belong to the cache placement set $S$, there is at least one content $n', n'\in S$ is not cached at edge servers. In this case, if we replace all caching $n$ with $n'$, we must acquire a higher reward. Hence, $i^*$ is not the optimal cache placement solution.
\end{proof}

The final cache placement solution is chosen from the cache placement set rather than the content set, which reduces the search complexity.
For each content, the edge server $m$ calculates the expected reward of caching it. The expected reward of content $n$ is the summary of the expected reward of content $n$ in different sub-regions. We assume that $R_m$ is divided to $P$ sub-regions, where $R_{m,p}$ denotes the $p$th sub-region of $R_m$.  We assume that $k_{m,p,n}$ edge servers share the sub-region $R_{m,p}$ and cache content $n$. If $k_{m,p,n}$ edge servers cache the same content $n$, the satisfied user number of caching $n$ is evenly assigned to each of them. The expected reward $\bar r_{m,n}$ is defined as
\begin{equation}
    \hat r_{m,n} = \sum_{p = 1}^P \frac{R_{m,p} *\mu(\hat\theta)*\hat p_n}{k_{m,p,n}}.
\end{equation}

After calculating the expected reward of $N$ contents, according to Problem 1, the edge server chooses $K$ content with the highest expected reward as the cache placement solution. The detail of decentralized cache placement solution is shown as Algorithm 2.}
\begin{algorithm}[]
\caption{Decentralized Cache Placement Based on Extended MAB of $m$ } 
\begin{algorithmic}
\STATE Input the amount of contents $N$, the cache size $K$; the amount of cooperation edge servers $M$; the sub-region $R_{m,p}$, for each content $\bar p_n = 0$; $t = 1$; the batch size $B$;  the primary state slot as $m^*$;
\WHILE{$t\ge1$}
\IF{$m$ is in primary time slot}
\STATE b = 0;
\WHILE{$b<B$}
\IF{$\log_2(t)\in\mathbb{N}$}
\STATE Select combination $\mathbbm{I}_{m,t}$ randomly;
\ELSE
\STATE Receive the cache placement solution from the adjacent edge servers;
\STATE Calculate the estimated reward $\hat r_{m,n}$ for content in $S$ with (20);
\STATE Choose the $K$ content with the highest expected reward as cache placement solution $\mathbbm{I}_{m,t}$;
\ENDIF
\STATE Receive the satisfied user number as reward $r_{m,t}$;
\STATE Calculate $ \bar X_{\mathbbm{I}_{m,t}}$ with (8);
\STATE Calculate $\hat\theta_{m}$ with (10);
\STATE Calculate $\hat p_{\mathbbm{I}_{m,t}}$ with (11);
\STATE Calculate $\hat p_n$ with (19);
\STATE $M_{{\mathbbm{I}_{m,t}}}(t) = M_{{\mathbbm{I}_{m,t}}}(t-1) +1$;
\STATE $b = b+1$;
\ENDWHILE
\STATE Send the cache information to other edge servers;
\ELSE
\STATE Keep the same cache placement solution as the last time slot;
\ENDIF
\ENDWHILE
\end{algorithmic}
\end{algorithm}

The regret analysis of each agent in the cooperative edge server MEC network is given as Proposition 3.

\begin{prop}
Until time $T$, the regret in decentralized Extended MAB based cache placement solution is $ (\frac{\log_2 T}{T} B*K
+2\exp\big(-2({{N-1}\choose{K-1}}\frac{\sigma_1}{\bar D_{1}})^{2\gamma_1} T\big)
+2N\exp \big(-2(\frac{\sigma_2}{D_2}\mu(\theta^*)^{2\gamma_2} T) \big))*gap_{max}$, where $D_1$, $D_2$, $\gamma_1$, $\gamma_2$ and $\sigma_1$, $\sigma_2$ are parameters satisfying $D_1>0$,$D_2>0$, $0<\gamma_1<1$,$0<\gamma_2<1$, $\sigma_1>0$, $\sigma_2>0$.
\end{prop}

\begin{proof}

We define the optimal cache placement solution of edge server $m$ is $\mathbbm{I}_{m^*}, \mathbbm{I}_{m^*} = \left\{i_{m,1},...,i_{m,K}\right\}$. The actual cache placement solution of $m$ is defined as $\mathbbm{I}_{m_a}, \mathbbm{I}_{m_a} = \left\{i_{m_a,1},...,i_{m_a,K}\right\}$. The accumulated regret of collaborative cache placement solution until $T$ is written as
\begin{equation}
    Reg(T) =  \frac{1}{T}\sum_{t=1}^{T}\sum_{m = 1}^M r_{\mathbbm{I}_{m^*}} - r_{\mathbbm{I}_{m_a}}.
\end{equation}

The accumulated regret is summary of accumulated regret of random cache placement and accumulated regret of estimated optimal cache placement.

As we mentioned before, the random cache placement solution only happens when $\log_2(t)\in\mathbb{N}$. Given the highest gap $gap_{max}$ of reward between the optimal cache placement solution and an arbitrary cache placement solution, the regret in the random cache can be written as
\begin{equation}
\begin{split}
     Reg_{1}(T)&\leq\frac{1}{T}\sum_{t = 1}^{T}\sum_{k = 1}^K
     \textbf{1}({\log_2(t)\in\mathbb{N}})B*gap_{max}\\
     &= \frac{\log_2 T}{T}*B*gap_{max}*K.
\end{split}
\end{equation}

The regret of estimated optimal cache placement is represented as
\begin{equation}
\begin{split}
     Reg_{2}(T)\leq&(2\exp\big(-2({{N-1}\choose{K-1}}\frac{\sigma_1}{\bar D_{1}})^{2\gamma_1} T\big)\\
     +&2N\exp \big(-2(\frac{\sigma_2}{D_2}\mu(\theta^*)^{2\gamma_2} T) \big))*gap_{max}.
\end{split}
\end{equation}

The proof of (23) is shown as appendix 3.

By summing up $Reg_{1}(T)$ and $Reg_{2}(T)$, Proposition 3 is derived.
\end{proof}

  {The complexity of the centralized cache placement solution  depends on the number of the macro-combinations. At each time slot, the central server updates the parameters of the user density and popularity of each macro-combination, then searches the macro-combination with the highest popularity as the cache placement solution. The complexity is $\mathcal{O}({{N}\choose{K}}^M)$.}

  {The complexity of the decentralized cache placement solution of edge server $m$ is determined not only by the number of contents, but the sub-regions overlapping with other edge servers. The complexity of updating parameters is the same as the individual scenario. In this case, the complexity is determined by the number of sub-regions $P$ and the maximal edge server number $k_{m,p,n}$ that share the same sub-region. The complexity is written as $\mathcal{O}(P*k_{m,p,n}*N)$.   }

\section{Simulation results}

{To evaluate the performance of our proposed algorithms, we compare them with both intelligent cache placement solutions and common cache placement solutions. We choose the least frequently used (LFU) cache placement solution and least recently used (LRU) cache placement solution as the common solution. LFU solution replaces the content with the shortest requested time. LRU solution considers replacing the least recently used content in the cache space.  We use the UCB bandit based solution, which is proposed in \cite{Blasco_2014}, and $\epsilon$-greedy bandit as the intelligent baselines. The $\epsilon$-greedy bandit lets the agent choose the arms randomly with probability $1-\epsilon$ and choose the best arm with probability $\epsilon$.  Specifically, we set $\epsilon = 0.95$ in the $\epsilon$-greedy bandit. The UCB bandit based solution focuses not only on the reward but also on the exploration duration of each arm. The UCB bandit based solution adjusts the trade-off between exploration and exploitation according to the accumulated reward and the time required for playing different arms. Furthermore,  we refer to the collaborative cache placement method in \cite{Xu_2020} as a baseline to evaluate the performance of the cooperative server scenario. The collaborative cache placement method also discusses how the overlapped region influences the performance of cache placement and proposes a multi-agent cache placement solution. Different from the solution in this paper, \cite{Xu_2020} assumes that the user number and the location of each user are pre-defined.

We choose three different metrics to compare the performance of the proposed cache placement solution, which are the accumulated regret,  convergence time, and average satisfying user number. For the first and second metrics, we mainly discuss the performance of the proposed Extended MAB. With the third metric, we show the effects of using different cache placement solutions with different settings.}
\subsection{Performance in individual server scenario}

{To evaluate the performance of the Extended MAB cache placement in the individual edge server scenario, we first compare the performance with different content numbers. We choose $N = 5$ and $N = 10$ and that the cache capacity of the edge server is $2$. The batch $B$ is set as $20$.
We define the radius of the serving region of each edge server as $5$ and the parameter of user density as $5$.  The popularity of each content follows the Zipf distribution. 

\begin{figure}[h]
\centering
\includegraphics[width=0.5\textwidth]{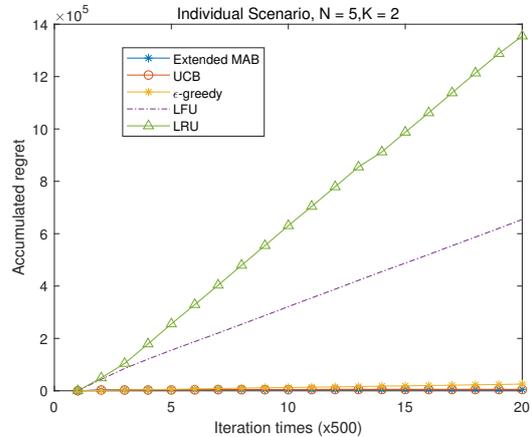}
\caption{The accumulated regret in the individual edge scenario with $N = 5, k = 2$.}
\label{singleagentM5res}
\end{figure}

\begin{figure}[h]
\centering
\includegraphics[width=0.5\textwidth]{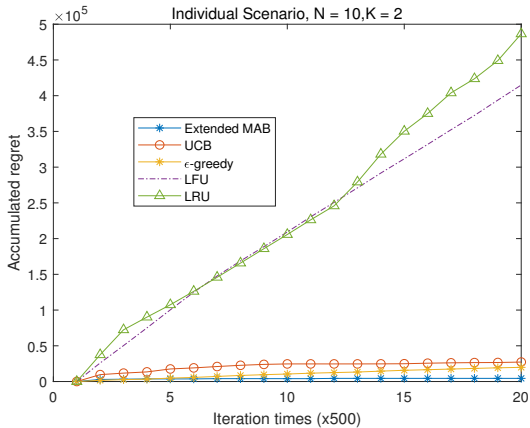}
\caption{The accumulated regret in the individual edge scenario with $N = 10, k = 2$.}
\label{singleagentM10res}
\end{figure}
Fig.~\ref{singleagentM5res} and Fig.~\ref{singleagentM10res} presents the accumulated regret of different cache placement solution with $N = 5$ and $N = 10$. The accumulated regret reflects the performance gap between the practical cache placement policy and the optimal cache placement policy. According to  Fig.~\ref{singleagentM5res} and Fig.~\ref{singleagentM10res}, we could see that the Extended MAB cache placement solution could achieve the optimal cache placement solution when it learns the environment well. According to Fig.~\ref{singleagentM5res} and Fig.~\ref{singleagentM10res}, the Extended MAB based solution has minimal accumulated regret. The LRU method shows an unstable performance as the LRU policy does not depend on the accumulated experience.

\begin{figure}[h]
\centering
\includegraphics[width=0.5\textwidth]{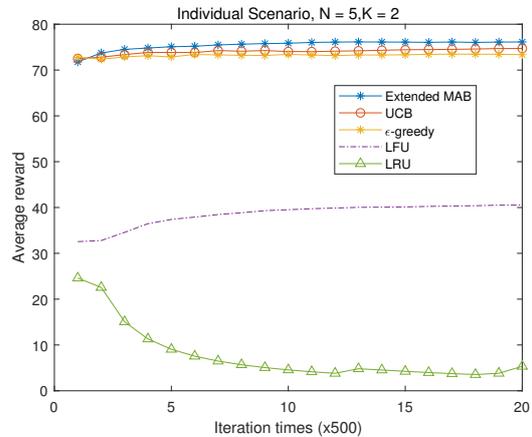}
\caption{The average satisfied user number in the individual edge scenario with $N = 5, k = 2$.}
\label{singleagentM5rewardres}
\end{figure}

\begin{figure}[h]
\centering
\includegraphics[width=0.5\textwidth]{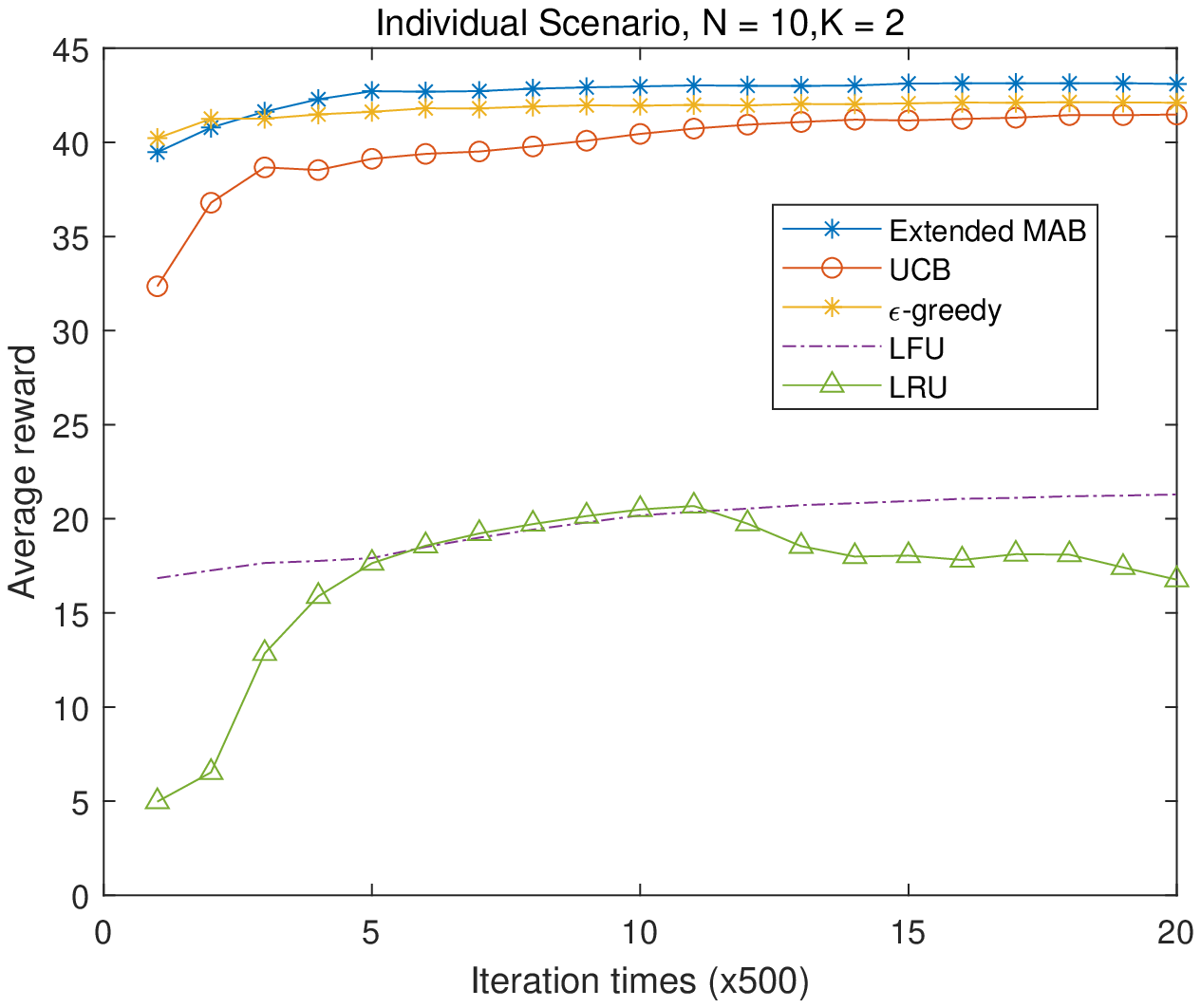}
\caption{The average satisfied user number in the scenario with $N = 10, k = 2$.}
\label{singleagentM10rewardres}
\end{figure}

Fig.~\ref{singleagentM5rewardres} and Fig.~\ref{singleagentM10rewardres} show the average satisfied user number by the cache placement with different solutions. The proposed algorithm has a better average reward and a faster convergence behavior than those of the  others. According to  Fig.~\ref{singleagentM5rewardres} and Fig.~\ref{singleagentM10rewardres}, we could see that with the experience knowledge, the learning based algorithms show better performance than the general cache solution. With the accurate estimation of the content popularity, the edge server could choose the optimal cache placement solution.

Moreover, we compare the accuracy of the parameter estimation among the proposed algorithm and the baselines. Table~\ref{convergencetime} illustrates the user density estimation accuracy under different numbers of iterations.} Table~\ref{convergencetime} presents the accuracy of estimating the parameter in different iteration times. It is observed that the Extended MAB solution can estimate the parameters more accurately and faster, which leads to a more accurate search of the optimal cache placement solution.
\begin{table} [h]
\caption{User density parameter estimation accuracy of different algorithms}
\centering
\begin{tabularx}{8.5cm}{llll}
\hline
Number of iterations & 4000  & 8000 & 12000 \\
\hline
$\epsilon$-greedy  & 0.4889 & 0.1247 & 0.1245 \\
UCB & 0.4694 & 0.0975 & 0.0895 \\
Proposed algorithm & 0.1817 & 0.0027 & 0.0053 \\
\hline
\end{tabularx}
\label{convergencetime}
\end{table}

\subsection{Performance in cooperative server scenario}

{{The centralized cache placement solution could be regarded as the large-scale individual cache placement solution, which has the similar performance as the last subsection. In this section, we only evaluate the decentralized cache placement solution.  It is noted that we introduce a new baseline from \cite{Xu_2020} to verify the availability of the decentralized Extended MAB based cache placement solution. In the following presentation, we call the baseline as Collab MAB. In the cooperative edge server scenario, we assume that the number of overlapping edge servers is 2 and 3 in the following experiments.  We set the content number as $N = 10$ and $N = 20$ and the cache size as 3 and 5 for each edge server. The overlap size between two edge servers is the half size of the serving region. We first give the accumulated regret and the average reward when $M = 2$.
The accumulated regret when $N = 10$ and $N = 20$ are shown as Fig.~\ref{multiagentC10} and Fig.~\ref{multiagentC20}.
\begin{figure}[h]
\centering
\includegraphics[width=0.5\textwidth]{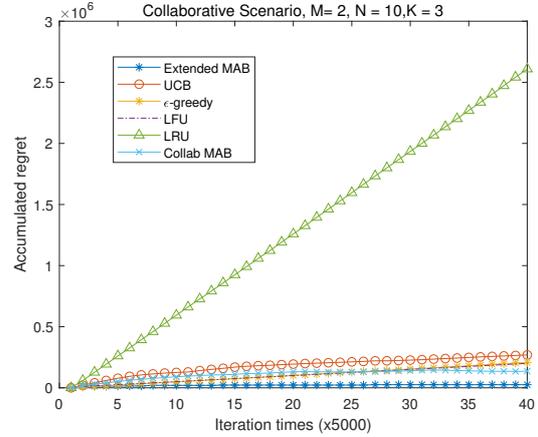}
\caption{The accumulated regret in the collaborative edge scenario with $M = 2, N = 10, k = 3$.}
\label{multiagentC10}
\end{figure}

\begin{figure}[h]
\centering
\includegraphics[width=0.5\textwidth]{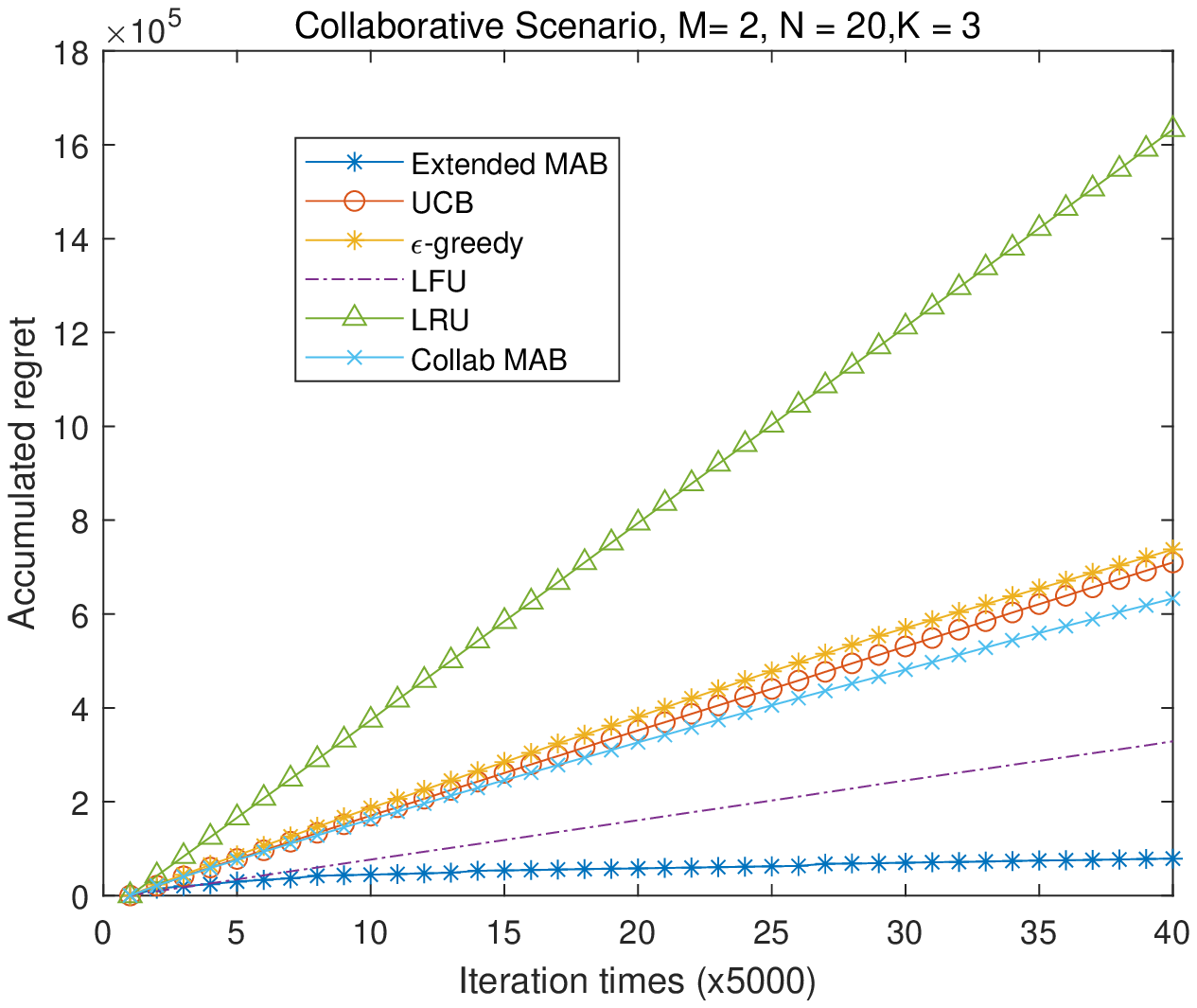}
\caption{The accumulated regret in the collaborative edge scenario with $M = 2, N = 20, k = 3$.}
\label{multiagentC20}
\end{figure}

According to Fig.~\ref{multiagentC10} and Fig.~\ref{multiagentC20}, the proposed Extended MAB solution has the minimum accumulated reward. Moreover, the UCB solution, $\epsilon$-greedy solution, and the Collab MAB solution have a better performance when the number of content is small. With the increase of the content number, these three baselines do not share the global parameter, i.e., the user density, which hinders the cache placement learn the environment and derive a better policy. Meantime, because of the influence of the overlapped region, these three baselines do not estimate the actual reward accurately, which further increases the accumulated regret.

\begin{figure}[h]
\centering
\includegraphics[width=0.5\textwidth]{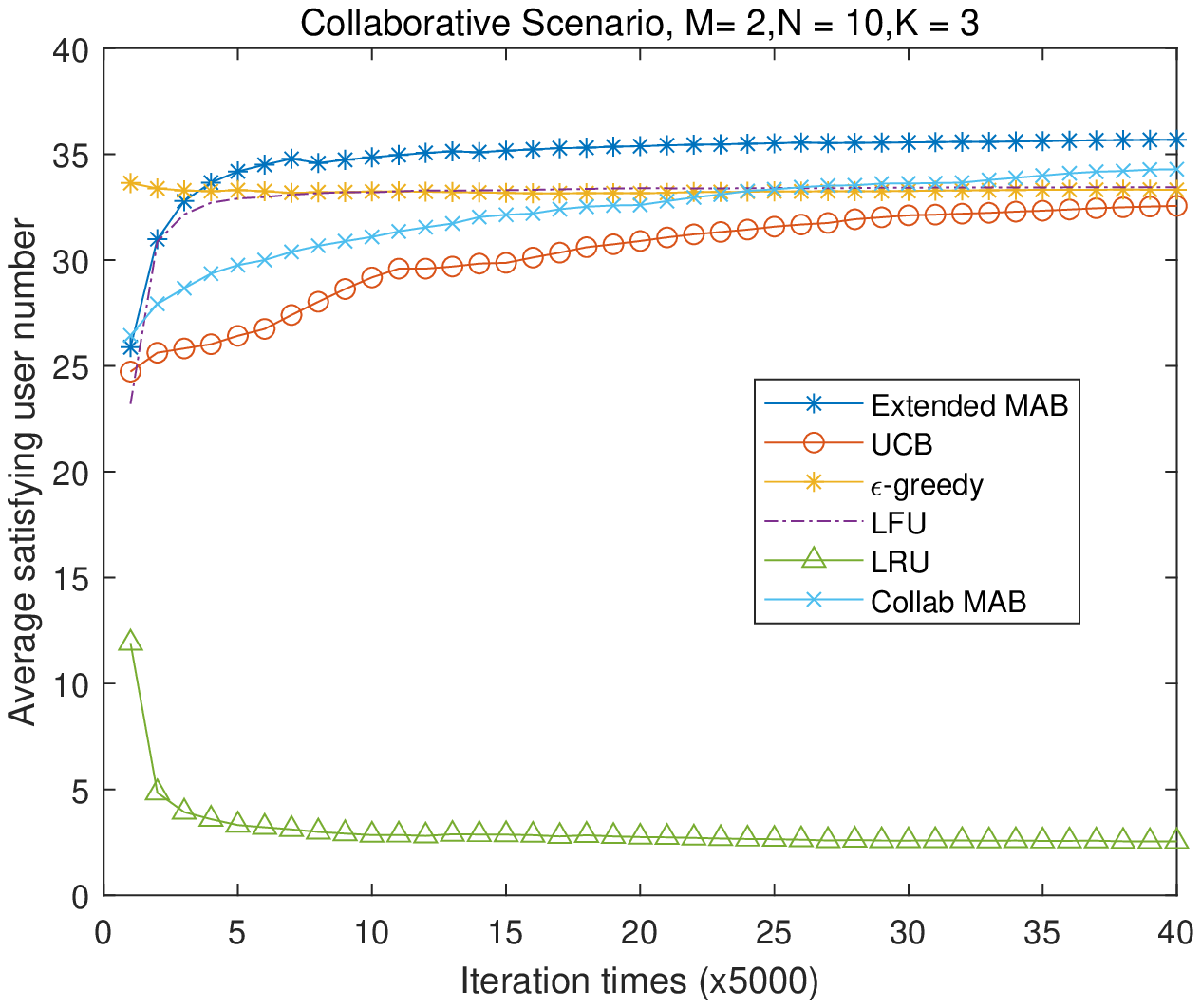}
\caption{The average satisfied user number in the collaborative edge scenario with $M = 2, N = 10, k = 3$.}
\label{multiagentC10rev}
\end{figure}

\begin{figure}[h]
\centering
\includegraphics[width=0.5\textwidth]{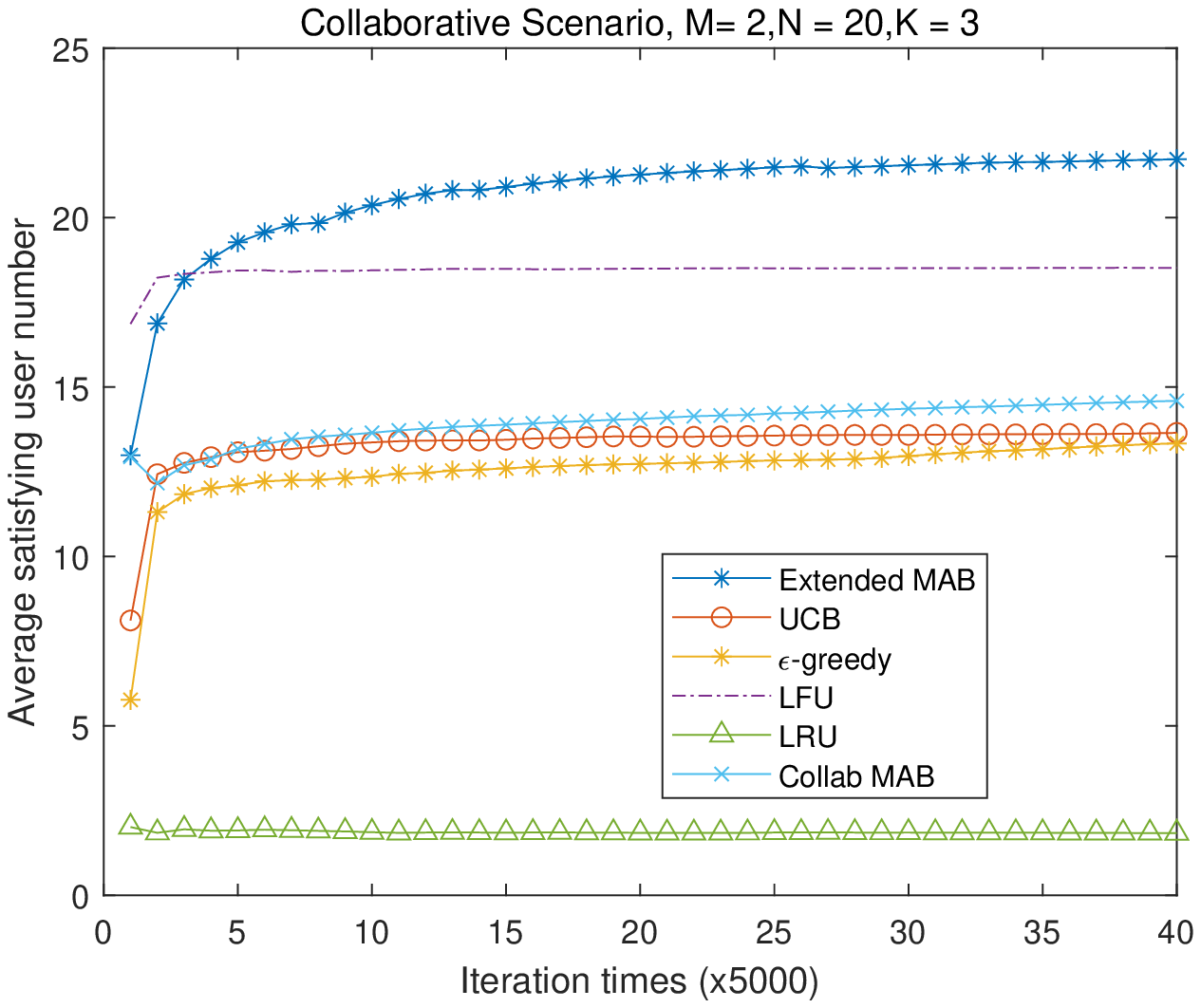}
\caption{The average satisfied user number in the collaborative edge scenario with $M = 2, N = 20, k = 3$.}
\label{multiagentC20rev}
\end{figure}

The average number of satisfied user are shown as Fig.~\ref{multiagentC10rev} and Fig.~\ref{multiagentC20rev}.
According to Fig.~\ref{multiagentC10rev} and Fig.~\ref{multiagentC20rev}, the average satisfied user number improves when the user number increases. With the same Zipf distribution parameter, if the content number increases, the popularity of each content tends to be evenly distributed. In this case, the proposed Extend MAB cache placement solution could adjust the policy based on the overlapped region, which leads to better performance.

Since the popularity of each content influences the optimal cache placement solution, we also adjust the Zipf parameter to evaluate the performance with different popularity of contents. The Zipf parameter is chosen as $\left\{0,0.5,1,1.5\right\}$ to further evaluate the performance of the average satisfied user number. According to Fig.~\ref{zipfinf}, the propose Extended MAB always shows the best performance with the change of Zipf parameter.

\begin{figure}[h]
\centering
\includegraphics[width=0.5\textwidth]{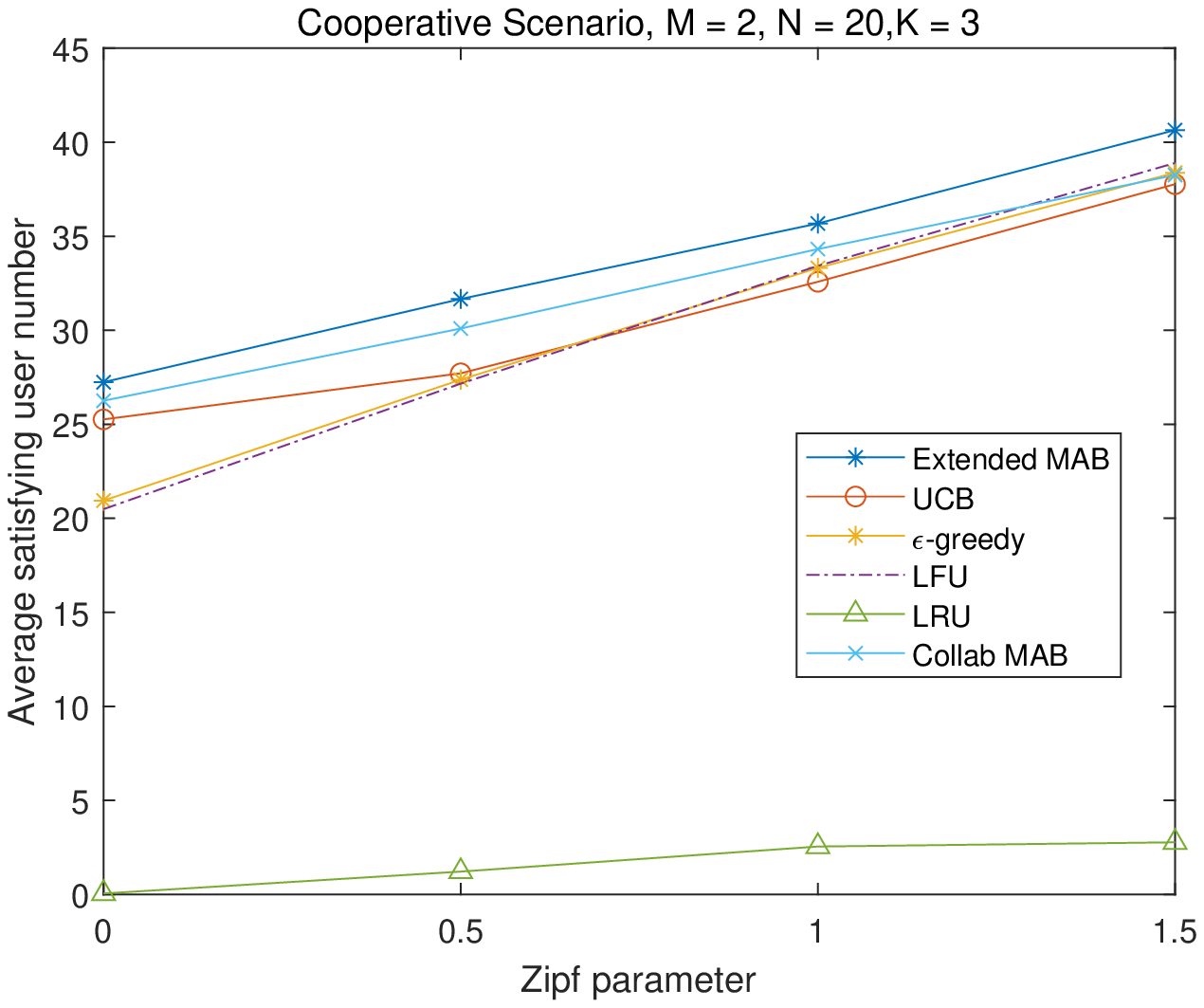}
\caption{The average satisfied user number with different zipf parameter.}
\label{zipfinf}
\end{figure}

Then, we set the number of edge servers as $3$. The radius of the serving region is set as $3$. Any two of them share a overlapped region with area size $2.2$ and all of them share a overlapped region with the area size $0.8$. We consider $20$ contents in the environment. The Zipf parameter is defined as $0.5$. The accumulated regret and average reward are shown as Fig.~\ref{multiagentC20reg} and Fig.~\ref{multiagentC20rev3u}.

\begin{figure}[h]
\centering
\includegraphics[width=0.5\textwidth]{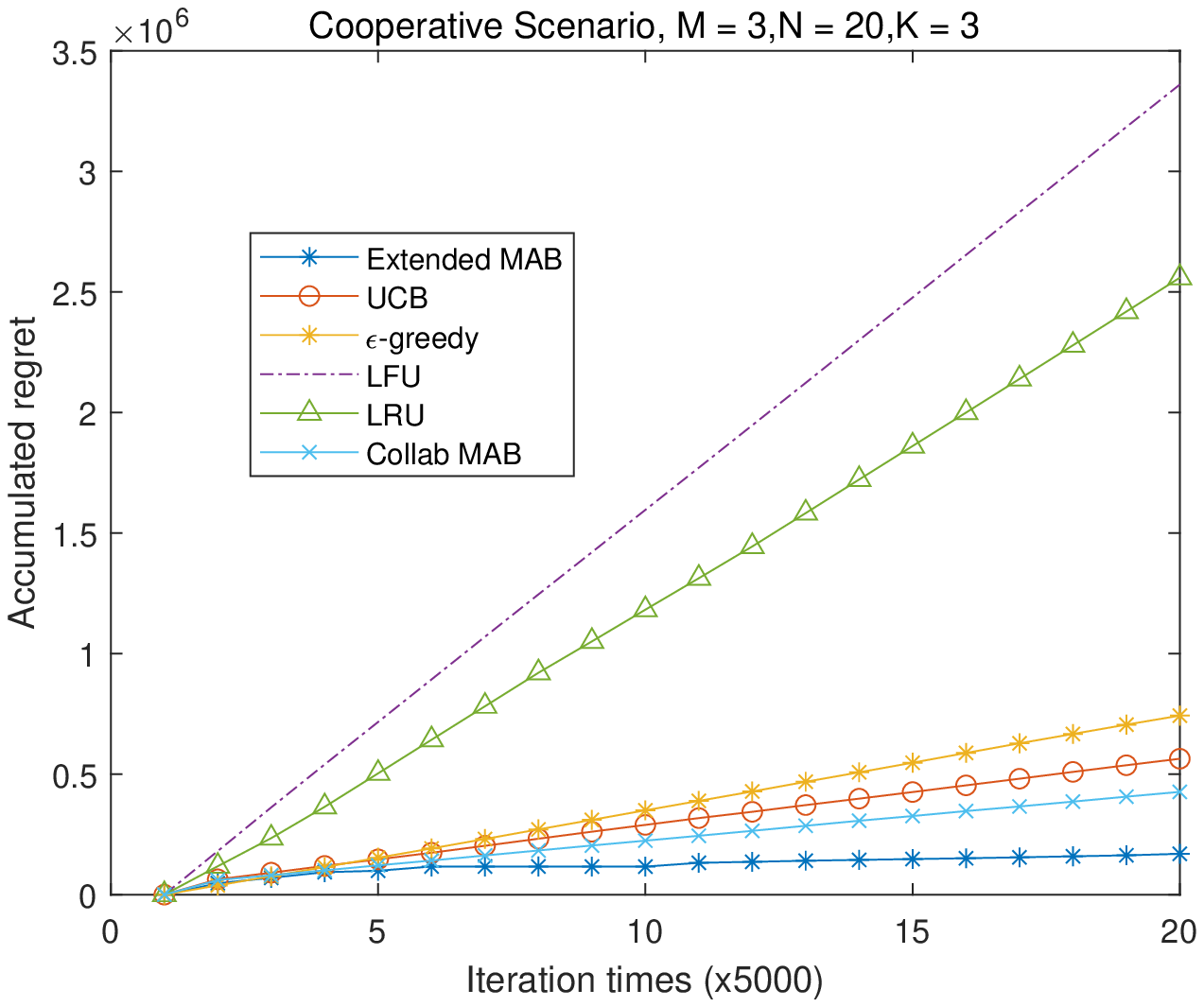}
\caption{The accumulated regret in the collaborative edge scenario with $M = 3, N = 10, k = 3$.}
\label{multiagentC20reg}
\end{figure}

\begin{figure}[h]
\centering
\includegraphics[width=0.5\textwidth]{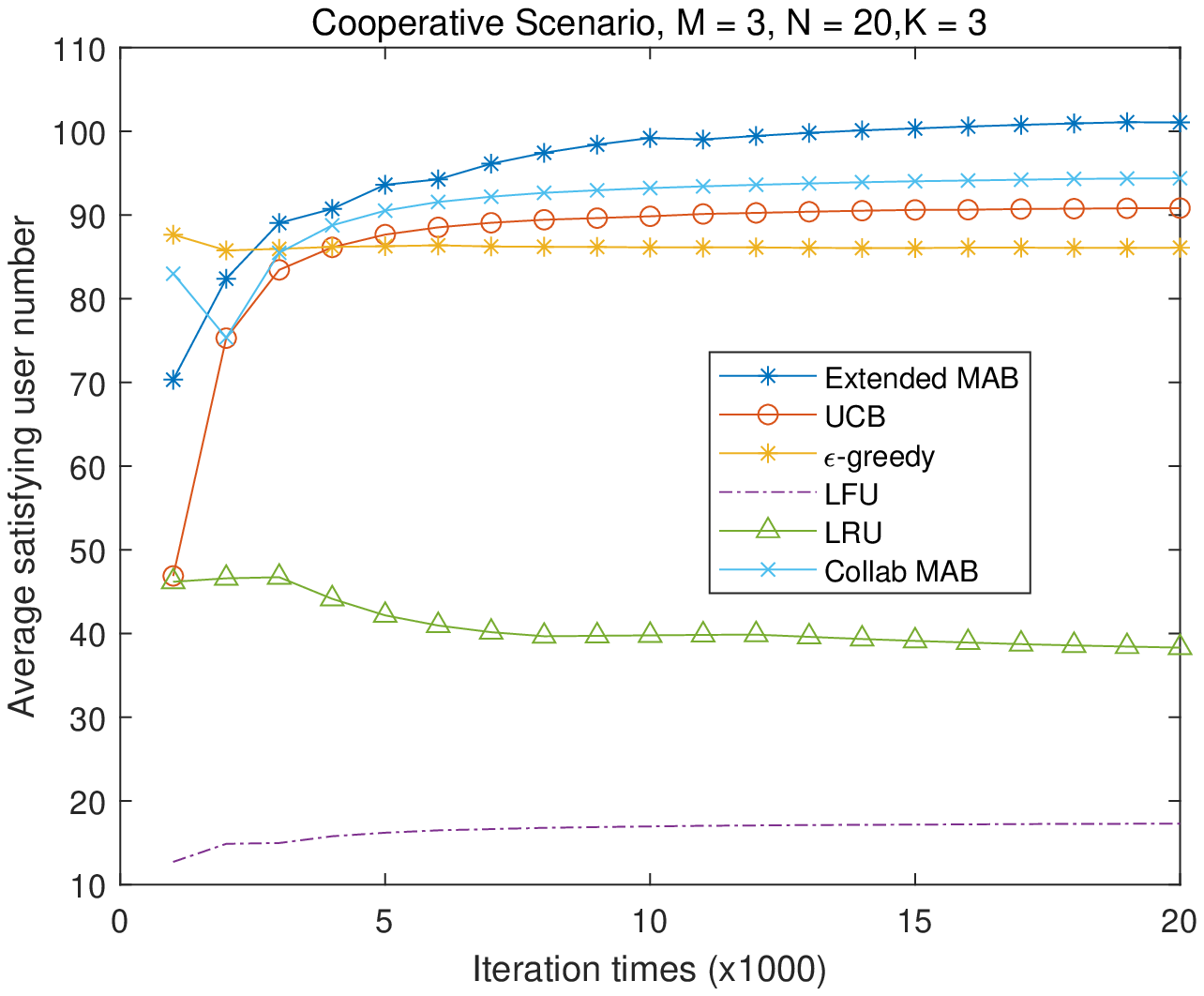}
\caption{The average satisfied user number in the collaborative edge scenario with $M = 3, N = 20, k = 3$.}
\label{multiagentC20rev3u}
\end{figure}

According to the Fig.~\ref{multiagentC20reg}, the accumulated regret has the same tendency as the scenario of $M = 2$. If the edge servers have overlapped region, the proposed cache placement could show better performance than all the baselines. The average satisfied user number in Fig.~\ref{multiagentC20rev} further verifies the performance improvement. The Extended MAB based cache placement solution has the highest average satisfied user number when the overlapped region is considered.

\begin{figure}[h]
\centering
\includegraphics[width=0.5\textwidth]{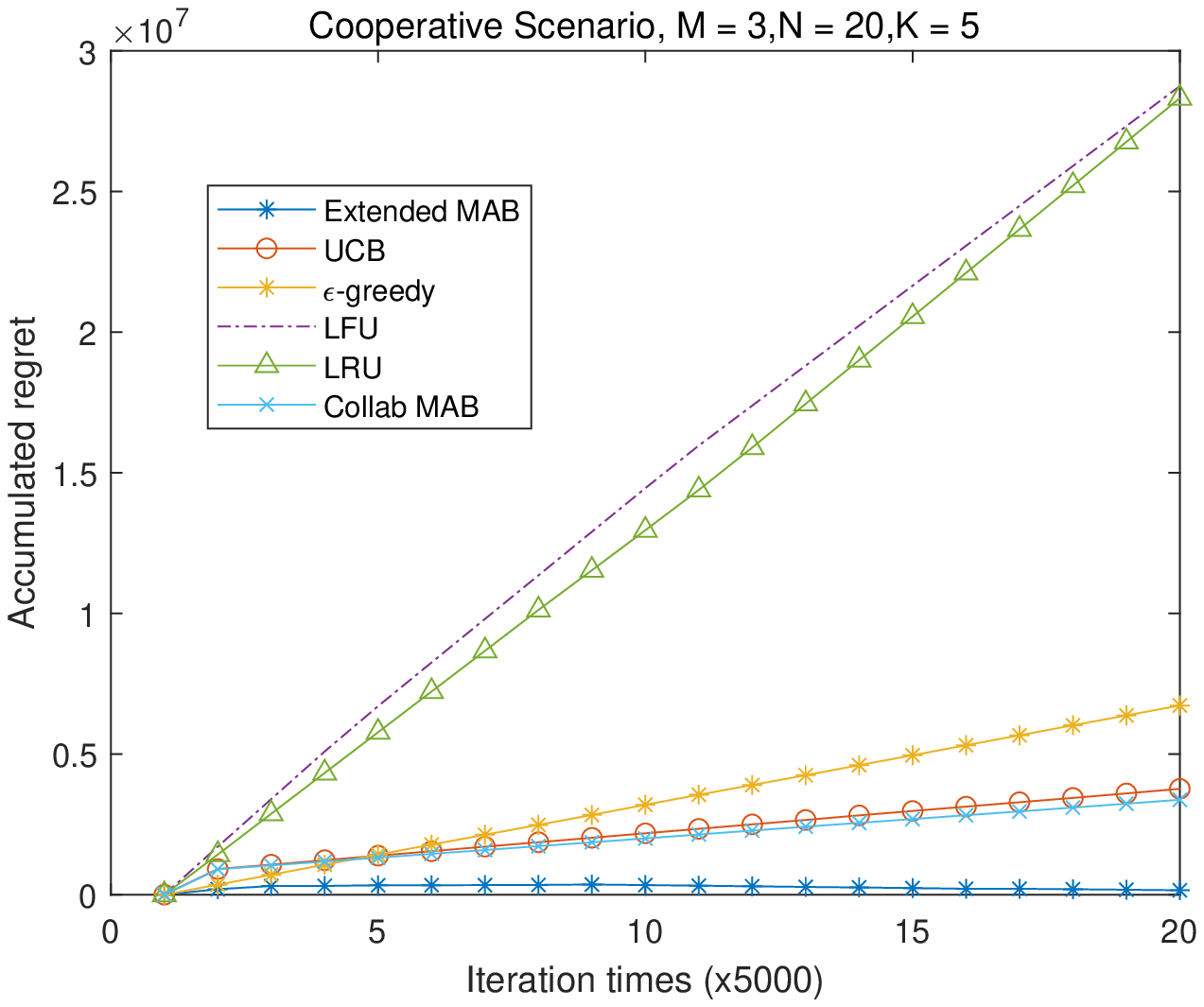}
\caption{The accumulated regret in the collaborative edge scenario with $M = 3, N = 20, k = 5$.}
\label{multiagentC205reg}
\end{figure}

\begin{figure}[h]
\centering
\includegraphics[width=0.5\textwidth]{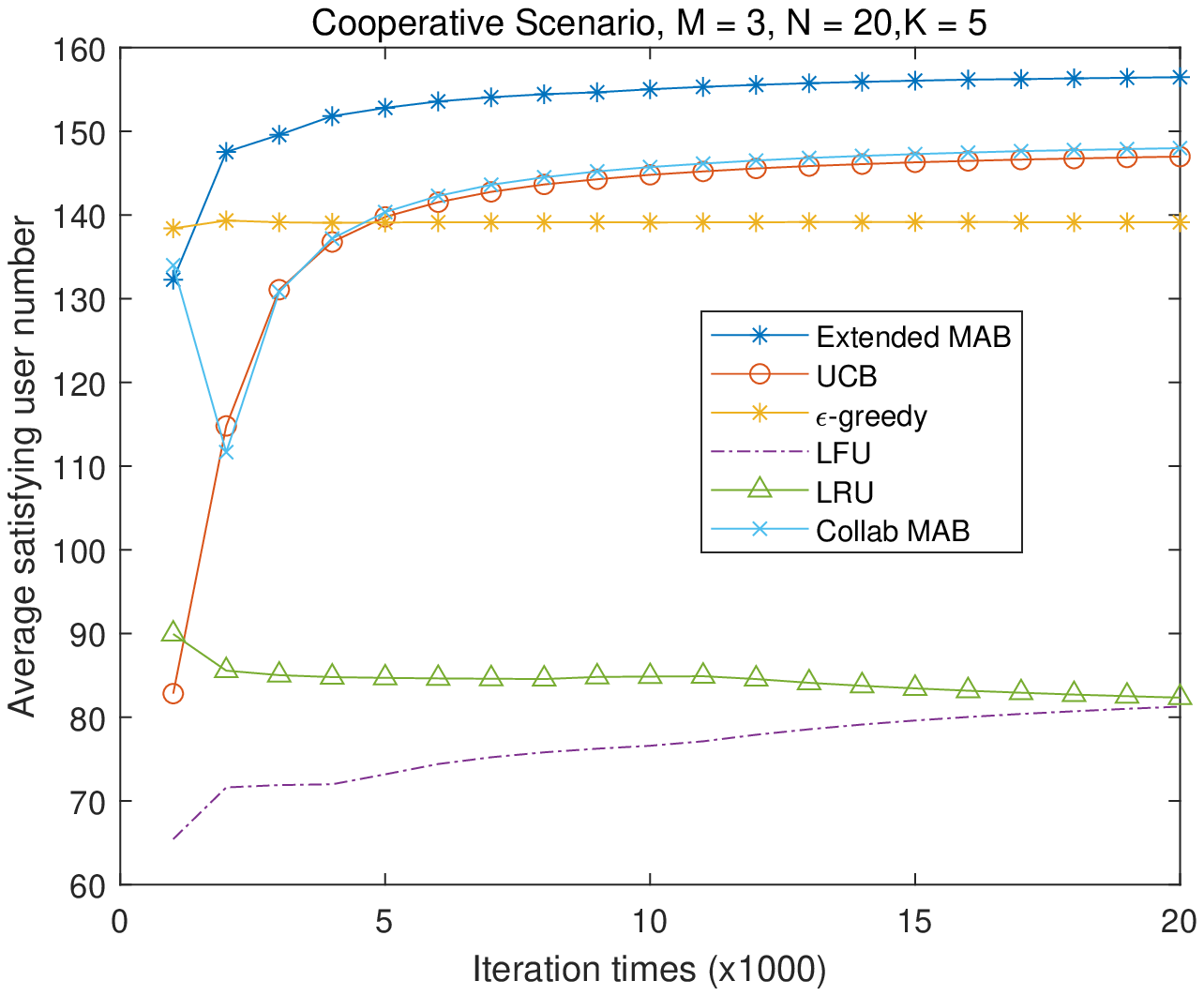}
\caption{The average satisfied user number in the collaborative edge scenario with $M = 3, N = 20, k = 5$.}
\label{multiagentC205rev}
\end{figure}

Fig.~\ref{multiagentC205reg} and Fig.~\ref{multiagentC205rev} indicate the accumulated regret and average reward when we set the content number as $20$ and cache size as $5$. Even though we increase the number of content and  further increase the number of cache combinations for each edge server, we still could get the same conclusion that the proposed Extended MAB based solution performs better if the overlapped region is considered in the cache placement solution.

 }}

\section{Conclusions}
In this paper, we study cache placement in the practical MEC network with unknown user behaviors. We develop an extended MAB based on the standard MAB and GGB to formulate the problem. To solve the problem, we first give a cache placement strategy for the individual edge server scenario, where edge servers share a non-overlapping serving region. After that, the joint cache placement is proposed for the cooperative edge server scenario with overlapping serving regions. To reduce the complexity of the proposed centralized algorithm, a decentralized extended MAB is proposed. The simulation results show that the proposed algorithm has the best performance compared to the baseline algorithms in different situations.

\begin{appendices}
      \section{Proof of Equation (14)}\label{ap1}

When $\log_2(t)\in\mathbb{N}$ is satisfied, the algorithm chooses the cache placement randomly. It is noted that as the reward from a certain action is limited, the regret between the best cache placement and other cache placement is limited. we denote the upper bound of the regret is denoted as $gap_{max}$. Hence, after a batch of random action, the regret of the exploration action is
\begin{equation}
\begin{split}
     Reg(T)&=\frac{1}{T}\sum_{t = 1}^{T} \textbf{1}({\log_2(t)\in\mathbb{N}})B*  (a_k\mu(\theta) - X_{I_{k,t}})\\
            &\leq \frac{\log_2 T}{T}* B*gap_{max}*K
            \end{split}
\end{equation}
where $\textbf{1}({\log_2(t)\in\mathbb{N}}) = 1$ if ${\log_2(t)\in\mathbb{N}}$ is satisfied, otherwise, $\textbf{1}({\log_2(t)\in\mathbb{N}}) = 0$.

      \section{Proof of Equation 15 }
      \label{ap2}
      It is noted that if the unknown parameters are estimated correctly, the practical cache placement is the same as the best cache placement. The regret between the optimal cache placement and the practical cache placement is caused by the wrong estimation of the parameters. According to the analysis, we could derive that
\begin{equation}
\left\{I_{max}\neq I_{chosen} \right\}
\subset \left\{ \hat\theta\neq\theta\right\}\cup\left\{\left\{ \exists n, \hat p_n \neq p_n \right\}\cap\left\{ \hat\theta=\theta\right\} \right\}.
\end{equation}

Based on (25), we separately analyze the regret with different wrong parameters. We firstly discuss the probability of $\hat\theta\neq\theta$. If $\hat\theta\neq\theta$, there exists $\sigma_1\geq 0$ which conforms $|\hat\theta - \theta|\ge\sigma_1$. We  establish the relationship between the estimated parameter of the user density and the expected reward of each combination until $t$, that is, there exists $ D_1>0$ and $ 1>\gamma_1>0$ satisfying
\begin{equation}
\begin{split}
&Pr\left\{\hat\theta\neq\theta \right\} = Pr({|\hat\theta - \theta|\ge\sigma_1})\\
\leq& Pr( D_1|\frac{1}{{{N-1}\choose{K-1}}}(\bar X_1 + ... + \bar X_{n}) - \mu(\theta)|^{\gamma_1} \ge \sigma_1).
\end{split}
\end{equation}
We define $\bar X$ as a variable which satisfies $\bar X=\frac{1}{{{N-1}\choose{K-1}}} (\bar X_1 + ... + \bar X_{n})$. It could be easily concluded that the expectation of $\bar X$ is $\mu(\theta)$. Eq. (26) could be rewritten as
\begin{equation}
\begin{split}
&Pr( D_1|\frac{1}{{{N-1}\choose{K-1}}}(\bar X_1 + ... + \bar X_{n}) - \mu(\theta)|^{\gamma_1} \ge \sigma_1)\\
&=Pr(|\bar X - {{N-1}\choose{K-1}}\mu(\theta)|^{\gamma_1}\ge \big({{N-1}\choose{K-1}}\frac{\sigma_1}{ D_{1}}\big))\\
\overset{(1)}{\leq}&2\exp\big(-2({{N-1}\choose{K-1}}\frac{\sigma_1}{ D_{1}})^{2\gamma_1} T\big).
\end{split}
\end{equation}
where (1) succeeds due to  Hoeffding's inequality.

Next we determine the probability $Pr\left\{\exists n, \hat p_n \neq p_n|\hat\theta=\theta\right\}$ following the similar analysis. There exists $\sigma_2\geq 0$ satisfying
\begin{equation}
Pr\left\{\exists n, \hat p_n \neq p_n|\hat\theta=\theta\right\}\\
\leq\sum_{n=1}^{N}Pr\big( |\hat p_n -p_n|\ge\sigma_2|\hat\theta=\theta\big).
\end{equation}
Since we already discussed the situation of $\hat\theta=\theta$, we could remove the condition directly. Based on Assumption 1, there exists $ D_2>0$ and  $1>\gamma_2>0$ which satisfies
\begin{equation}
Pr\left\{\hat p_n \neq p_n|\hat\theta=\theta\right\}
\leq 2\exp \big(-2(\frac{\sigma_2}{D_2}p_n\mu(\theta)^{2\gamma_2} T) \big).
\end{equation}

\section{Proof of Equation (23)}\label{ap3}
When the edge server $m$ chooses the expected optimal cache placement solution, the regret  is
\begin{equation}
\begin{split}
Reg(T) &= gap_{max}*Pr(I_{max}\neq I_{chosen}) \\
\leq& gap_{max}*Pr(p_1\neq\hat p_1\cup...\cup p_N\neq\hat p_N\cup \theta\neq\hat \theta)\\
 \leq& gap_{max}*(Pr( \theta\neq\hat \theta)+\sum_{n=1}^N Pr(p_n\neq\hat p_n| \theta\neq\hat \theta)).
\end{split}
\end{equation}

According to Eq. (27) and Eq. (29), we have the following inequality regarding the expected regret of optimal cache placement policy
\begin{equation}
\begin{split}
&Reg(T) \leq (2\exp\big(-2({{N-1}\choose{K-1}}\frac{\sigma_1}{ D_{1}})^{2\gamma_1} T\big)\\
+&2N\exp \big(-2(\frac{\sigma_2}{D_2}\mu(\theta^*)^{2\gamma_2} T) \big))*gap_{max}.
\end{split}
\end{equation}

We obtain the expected regret of each edge server in cooperative scenario given by
\begin{equation}
\begin{split}
Reg(T) \leq (\frac{\log_2 T}{T}*K
+2\exp\big(-2({{N-1}\choose{K-1}}\frac{\sigma_1}{ D_{1}})^{2\gamma_1} T\big) \\
+2N\exp \big(-2(\frac{\sigma_2}{D_2}\mu(\theta^*)^{2\gamma_2} T) \big))*gap_{max}.
\end{split}
\end{equation}
\end{appendices}

\end{document}